\documentclass[sigconf]{acmart}

\AtBeginDocument{%
	\providecommand\BibTeX{{%
			\normalfont B\kern-0.5em{\scshape i\kern-0.25em b}\kern-0.8em\TeX}}}

\copyrightyear{2023}
\acmYear{2023}
\setcopyright{acmcopyright}\acmConference[WSDM '23]{Proceedings of the Sixteenth ACM International Conference on Web Search and Data Mining}{February 27-March 3, 2023}{Singapore, Singapore}
\acmBooktitle{Proceedings of the Sixteenth ACM International Conference on Web Search and Data Mining (WSDM '23), February 27-March 3, 2023, Singapore, Singapore}
\acmPrice{15.00}
\acmDOI{10.1145/3539597.3570477}
\acmISBN{978-1-4503-9407-9/23/02}

\usepackage{multirow}
\usepackage{subfigure}
\usepackage{algorithm}
\usepackage{algorithmic}

\begin{document}
	
	\title{Unbiased Knowledge Distillation for Recommendation}
	
	\author{Gang Chen}
	\affiliation{%
		\institution{University of Science and Technology of China}
		\city{Hefei}
		\country{China}}
	\email{gchenx@mail.ustc.edu.cn}
	
	\author{Jiawei Chen}
	\authornote{Corresponding author}
	\affiliation{%
		\institution{Zhejiang University}
		\city{Hangzhou}
		\country{China}
	}
	
	\email{sleepyhunt@zju.edu.cn}
	
	\author{Fuli Feng}
	\affiliation{%
		\institution{University of Science and Technology of China}
		\city{Hefei}
		\country{China}
	}
	\email{fulifeng93@gmail.com}
	
	\author{Sheng Zhou}
	\affiliation{%
		\institution{Zhejiang University}
		\city{Hangzhou}
		\country{China}
	}
	\email{zhousheng_zju@zju.edu.cn}

	\author{Xiangnan He}
	\authornotemark[1]
	\affiliation{%
		\institution{University of Science and Technology of China}
		\city{Hefei}
		\country{China}
	}
	\email{xiangnanhe@gmail.com}

	%%\renewcommand{\shortauthors}{Trovato and Tobin, et al.}
	
	\newcommand{\ie}{\emph{i.e., }}
	\newcommand{\eg}{\emph{e.g., }}
	\newcommand{\etal}{\emph{et al.}}
	\newcommand{\etc}{\emph{etc.}}
	\newcommand{\wrt}{\emph{w.r.t. }}
	\newcommand{\cf}{\emph{cf. }}
	\newcommand{\aka}{\emph{aka. }}

	\begin{abstract}
		
		As a promising solution for model compression, knowledge distillation (KD) has been applied in recommender systems (RS) to reduce inference latency. Traditional solutions first train a full teacher model from the training data, and then transfer its knowledge (\ie \textit{soft labels}) to supervise the learning of a compact student model. However, we find such a standard distillation paradigm would incur serious bias issue --- popular  items are more heavily recommended after the distillation. This effect prevents the student model from making accurate and fair recommendations, decreasing the effectiveness of RS.  %Therefore, it is essential to study this important but unexplored problem.

		In this work, we identify the origin of the bias in KD --- it roots in the biased soft labels from the teacher, and is further propagated and  intensified during the distillation. To rectify this, we propose a new KD method with a stratified distillation strategy. It first partitions items into multiple groups according to their popularity, and then extracts the ranking knowledge within each group to supervise the learning of the student. Our method is simple and teacher-agnostic --- it works on distillation stage without affecting the training of the teacher model. We conduct extensive theoretical and empirical studies to validate the effectiveness of our proposal. We release our code at: https://github.com/chengang95/UnKD.

	\end{abstract}
	
	
	

	\maketitle
	
	\section{Introduction}
	%Being able to provide personalized suggestions to each user, recommender systems (RS) have been widely used in countless online applications. 
	Recommender system (RS) has become increasingly important with the universalization of online personalized services. 
	%Recommender systems (RS) have been widely used in countless online applications.
	With the increasing scale of  items,  the trade-off between the accuracy and efficiency in modern RS cannot be ignored. A large model with numerous parameters has a high capacity, and thus is shown to have better  accuracy. However, its success requires heavy computational and memory costs, which would incur unacceptable latency during the inference phase, making it hard to be applied in real-time RS.

	To deal with such dilemma, \textit{knowledge distillation} (KD) has been applied in recommender system \cite{10.1145/3219819.3220021,8970837,10.1145/3442381.3449878}, with the purpose of reducing model size while maintaining model performance. KD  first trains a large teacher model from the training set, and then learns a small student model with the supervision from the \textit{soft labels} that are generated by the teacher. As the soft labels encode the knowledge learned by the teacher, the student can benefit more from it and achieve better performance than the student directly learning from the training data.
	
	Despite decent performance, we argue that the distillation is severely biased towards popular items. We make an empirical study of existing KDs on three benchmark recommendation datasets. The results are presented in Table \ref{tab:table_3}. The overall improvements of KDs mainly lie on the popular group, while the performance of the unpopular group drops significantly (22.4\% on average). This impressive result clearly reveal the severe bias issue in KDs, which is essential to be overcome. This negative effect will hinder the student model from completely understanding user preference. Worse still, it will decrease the level of the diversity and fairness in recommendations, heavily deteriorating user experience. 
	%Specifically, existing KDs are more friendly to popular items, reinforcing their performance, while systemically discriminate niche items and even hurt their performance to satisfy the needs of the mainstreams. This negative effect will hinder the student model from completely understanding user preference. Worse still, it will decrease the level of the diversity and fairness in recommendations, heavily deteriorating user experience. %Table \ref{tab:table_3} provide empirical evidences on three datasets: the overall improvements of KDs mainly lies on popular groups, while the scores (Recall@10) of unpopular groups are close to or even smaller than the baseline that is directly trained on training dataset. Therefore, it is essential to study the bias issue in KDs so that the niche groups can also sufficiently benefit from the teacher knowledge, and thus improve the overall recommendation quality.
	
	%To justify our argument, we make an empirical study of existing KDs on three benchmark recommendation datasets. The results are presented in Table \ref{tab:table_3}. The overall improvements of KDs mainly lie on the popular group, while the performance of the unpopular group drops significantly (22.4\% on average). This impressive result clearly reveal the severe bias issue in KDs, which is essential to be overcome.
	
	In view of this phenomenon, we first identify  the origin of the bias --- \textit{ the biased soft labels generated by the teacher} --- which is further intensified by the distillation process in training the student. %To be more specific, it is well-recognized that a recommendation model usually gives higher scores to popular items than their ideal values \cite{10.1145/3447548.3467289}. 
	Figure \ref{fig:image_6} provides the evidence of biased  teacher prediction, where we train a standard matrix factorization (MF) \cite{DBLP:journals/corr/abs-1205-2618}  and count the ratios of popular/unpopular items in the top-10 recommendation lists. As can be seen, the top-10 items with the largest scores are severely biased towards mainstream. Worse still, such bias would be inherited and amplified during the distillation. Existing KDs \cite{10.1145/3219819.3220021,8970837} usually simply consider higher-ranked items as positive and give them larger confidence weights. As a result, popular items would exert excessive contribution on student model training, causing the bias of the student. 
	
	%To better understand how KDs incur bias issue, we first analyze the output of the model using statistical methods. We train a 100-dimensional BPRMF model on three different datasets MovieLens, Apps and CiteULike, respectively. Then calculate the proportion of items from popular groups and items from unpopular groups in the recommendation results.  As shown in Figure \ref{fig:image_6}, the proportion of popular items in the recommendation results is much larger than that of unpopular items. It may root in biased data\cite{DBLP:journals/corr/abs-1909-06362}, factorization model \cite{DBLP:journals/corr/abs-1909-06362} or momentum-based optimizer \cite{tang2020longtailed}. Such effect creates bias in the teacher prediction, which will be further propagated and intensified into the student model through the distillation. This effect is undesirable but almost ubiquitous in recommender systems. We used the 10-dimensional BPRMF model as the student and then counted the performance improvement of different distillation methods, the results are shown in Table \ref{tab:table_3}. From the table, we can clearly find that the performance of all distillation methods improved in the popular group, while the performance decreased in the unpopular group. This is apparently due to the deviation of the teacher being passed on to the students through distillation, thus increasing the deviation of the students.

	Being aware of the origin of the distillation bias, now the question lies on how to eliminate this negative effect. A straightforward solution is to directly intervene into the training of the teacher model to generate unbiased soft labels, which however is difficult to achieve. On the one hand, the teacher bias may root in multiple factors, including but not limited to the momentum-based optimizer \cite{tang2020longtailed}, imperfect loss function \cite{10.1145/3404835.3462919}, and the factorization model architecture \cite{DBLP:journals/corr/abs-1909-06362}. Completely isolating bias from teacher is itself highly challenging. 
	%In our empirical study, we attempt to leverage debiasing strategies in the teacher model 
	%and discarding popularity completely may hurt the performance of the model (\ie user may buy just because a certain item is popular), which will reduce the upper bound of the distillation performance (\ie teacher performance). 
	On the other hand, a teacher-agnostic KD strategy is more desirable. In  practice, a large teacher model is usually deployed in a complex distributed system, and adjusting its training procedure is difficult objectively, not to mention the teacher can be an ensemble of multiple models. As such, in this work, we propose an \textbf{Un}biased \textbf{K}nowledge \textbf{D}istillation strategy (\textbf{UnKD}) that performs debiasing during the training of the student model. Specifically, UnKD resorts to a skillful popularity-aware distillation: it first partitions items into multiple groups according to their popularity, and then extracts the ranking knowledge among each group to supervise the learning of the student.
	On the basis of causal theory, we prove that such stratification strategy can almost block the causal effect from the teacher bias.
	%With borrowing causality theory, we proofed UnKD can almost properly offset the bias from the teacher.
	%On the basis of causal theory, we prove that \textit{group-wise learning} can effectively offset the bias from the teacher, and we propose UnKD accordingly.
	Remarkably, UnKD is simple and model-agnostic. We implement it on MF \cite{DBLP:journals/corr/abs-1205-2618} and LightGCN \cite{10.1145/3397271.3401063} to demonstrate effectiveness.
	
	To summarize, this work makes the following contributions:
	\begin{itemize}
		\item Revealing the bias issue of knowledge distillation in recommender systems.
		\item Proposing an unbiased teacher-agonistic knowledge distillation (UnKD) that extracts popularity-aware ranking knowledge to guide student learning.
		\item Conducting extensive experiments on three real datasets to demonstrate the superiority of UnKD over state-of-the-arts.
	\end{itemize}%
	
	The rest of the paper is organized as follows. In Section 2, we introduce the background of knowledge distillation. In Section 3, we provide causal view on bias issue in knowledge distillation and then detail our proposed UnKD. The experiments and discussions are presented in Section 4. Finally, we provide related work and conclusions in Section 5 and Section 6.

	\begin{table}[t]\small %
		\centering
		\caption{Performance (Recall@10) comparison of various knowledge distillation methods in terms of popular/unpopular items on three real-world datasets. We also report the relative improvements over the baseline (`Student') that is directly learned from the data. The experimental settings and the group partition are detailed in Section 4.}
		\setlength{\tabcolsep}{0.75mm}{
			\begin{tabular}{|c|c|c|c|c|c|c|}
				\hline
				\multicolumn{7}{|c|}{Movielens}\\
				\hline
				&Student&Teacher&RD\cite{10.1145/3219819.3220021}&CD\cite{8970837}&DERRD\cite{10.1145/3340531.3412005}&HTD\cite{10.1145/3447548.3467319} \\
				\hline
				Popular &0.2156&0.2565 &0.2237 &0.2258 &0.2315&0.2228  \\
				Group& ---& +18.97\%& +3.75\% &+4.73\%&+7.37\%&+3.33\%\\
				\hline
				Unpopular&0.0250&0.0517& 0.0242&0.0179 &0.0113&0.0187  \\
				Group& --- & +106.80\%&-3.20\% &-28.40\%  &-54.80\%  &-25.20\%  \\
				\hline
				\multicolumn{7}{|c|}{Apps}\\
				\hline
				&Student&Teacher&RD\cite{10.1145/3219819.3220021}&CD\cite{8970837}&DERRD\cite{10.1145/3340531.3412005}&HTD\cite{10.1145/3447548.3467319} \\
				\hline
				Popular &0.1031&0.1448 &0.1144&0.1212 &0.1058 &0.1195  \\
				Group&  ---& +40.44\%& +10.96\% &+17.55\%&+2.61\%&+15.90\%\\
				\hline
				Unpopular  &0.0109  &0.0164&0.0098&0.0090&0.0098 &0.0061    \\
				Group& --- & +50.45\%&-10.09\% &-17.43\%  &-10.09\%  &-44.03\%  \\
				\hline
				\multicolumn{7}{|c|}{CiteULike}\\
				\hline
				&Student&Teacher&RD\cite{10.1145/3219819.3220021}&CD\cite{8970837}&DERRD\cite{10.1145/3340531.3412005}&HTD\cite{10.1145/3447548.3467319} \\
				\hline
				Popular &0.0831&0.1294&0.0910&0.0887&0.0899 &0.0885 \\
				Group&---  & +55.71\%& +9.50\% &+6.73\%&+8.18\%&+6.49\%\\
				\hline
				Unpopular &0.0095& 0.0537&0.0085&0.0088 &0.0075  &0.0068   \\
				Group& --- & +465.26\%&-10.52\% &-7.36\%  &-21.05\%  &-28.42\%  \\
				\hline
		\end{tabular}}
		\label{tab:table_3}
	\end{table}%

\section{PRELIMINARIES}
	\begin{figure}[t]
	\centering
	
	\includegraphics[scale=0.2]{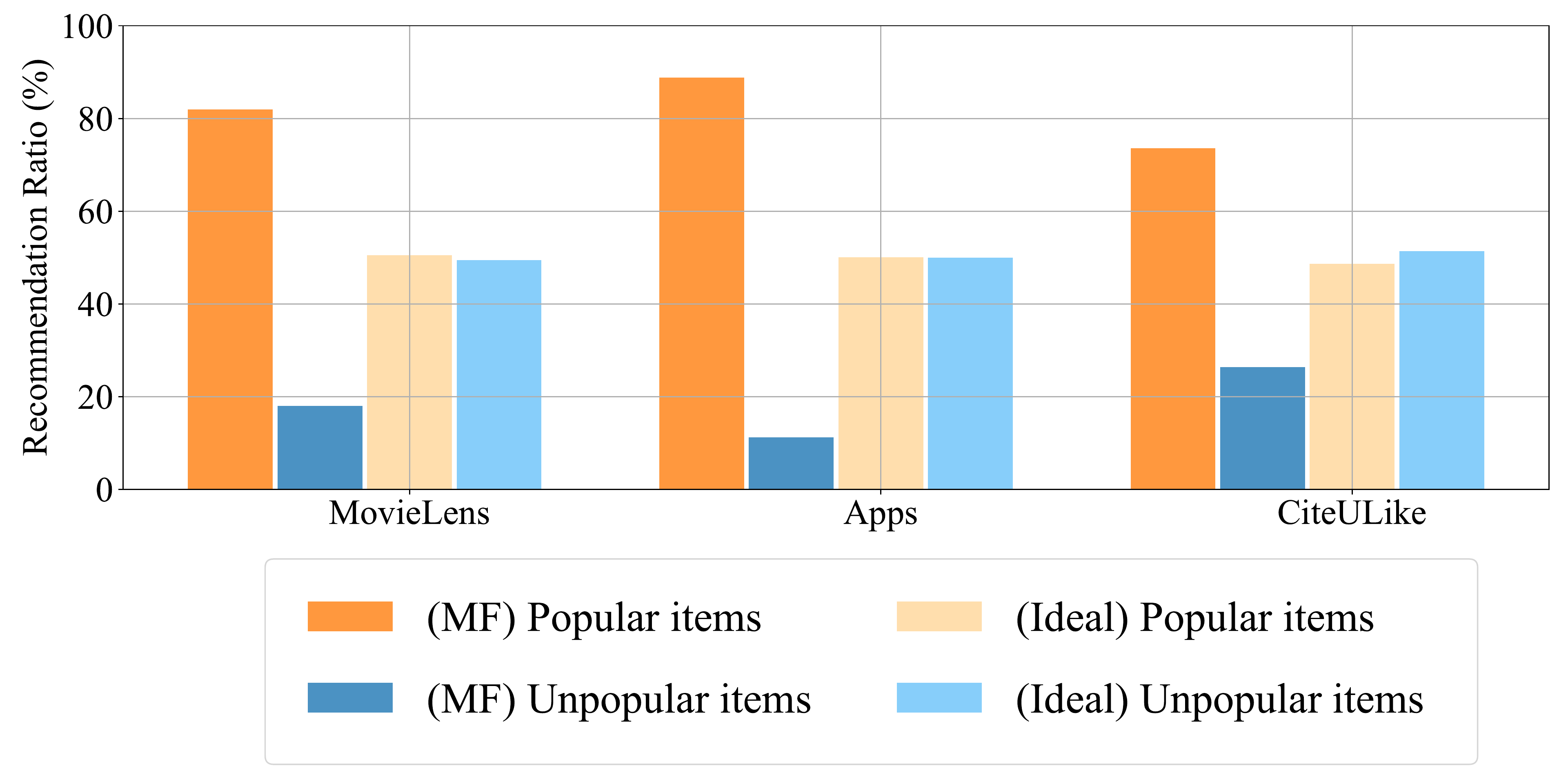}
	\caption{Ratios of popular/unpopular items in the top-10 recommendation lists from a MF model. We also present the ideal ratios from the test data for comparison.} 
	
	%Comparison of the proportion of items in the popular group and the unpopular group between the output of the MF model (Caculated) and the real test set (Ideal).}
\label{fig:image_6}
\end{figure}%
In this section, we first introduce the basic notations and formulate the recommendation task. We then provide the background of knowledge distillation. 

We use uppercase character (\eg$U$) to denote a random variable
and lowercase character (\eg $u$) to denote its specific value. We use
characters in calligraphic font (\eg$\mathcal{U}$) to represent the sample space
of the corresponding random variable. We use the notation $|*|$ for the size of the collection, \eg $| \mathcal{U} |$ denoting the size of $\mathcal{U}$.

\textbf{Recommendation Task}. Suppose we have a recommender system with a user set $\mathcal{U} =\{u_{1},... ,u_{m}\}$ and an item set $\mathcal{I} = \{i_{1},...,i_{n}\}$. The collected historical user-item feedback can be formulated as a matrix $R\in \{1,0\}^{m\times n}$, where each entry $r_{ui} \in R$ denotes whether a user $u$ has interacted with an item $i$. Given user $u$, $\mathcal{I}_{u}^{+} = \{i \in \mathcal{I}|r_{ui} = 1\}$ is the set of items with known positive feedback, and $\mathcal{I}_{u}^{-} =\mathcal{I}\backslash \mathcal{I}_{u}^{+}$ is the set of items with missing feedback \cite{4781121,10.1145/2911451.2911502,10.5555/3020488.3020521}. For each user, the goal of a recommender system is to find items that are most likely to be interacted.

\textbf{Knowledge Distillation}. Knowledge distillation \cite{Liu_2019_CVPR,NEURIPS2021_b9f35816,Hu_2021_CVPR} is a promising model compression technique that first trains a large teacher model and then transfers the knowledge from the teacher to the target compact student model. 
%A KD usually consists of two steps: 1) a large teacher model is trained with the training dataset; and 2) a compact student model is trained with additional knowledge from the teacher model. 
In RS, soft labels --- \ie the teacher predictions on the user-item interactions, are commonly-used for knowledge transfer. These KDs \cite{10.1145/3219819.3220021,8970837,10.1145/3340531.3412005} would create or sample the training instances according to soft labels for training a student model. As such, the quality of the soft labels lays a foundation of knowledge distillation. The distillation process is shown in Figure \ref{causal:causal_4}(a). It is worth to mention that there is work HTD \cite{10.1145/3447548.3467319} that utilizes teacher embeddings rather than soft labels for knowledge distillation. However, we also observe serious bias issue in HTD (Table \ref{tab:table_3}). In fact, our analyses are also suitable for this method, if we simply replace the term `Soft labels' with `Teacher embeddings'. 

\begin{figure}[t]
	\centering
	
	\subfigure[Generative process of the soft label.]{
		\begin{minipage}[t]{0.5\linewidth}
			\centering
			\includegraphics[scale=0.7]{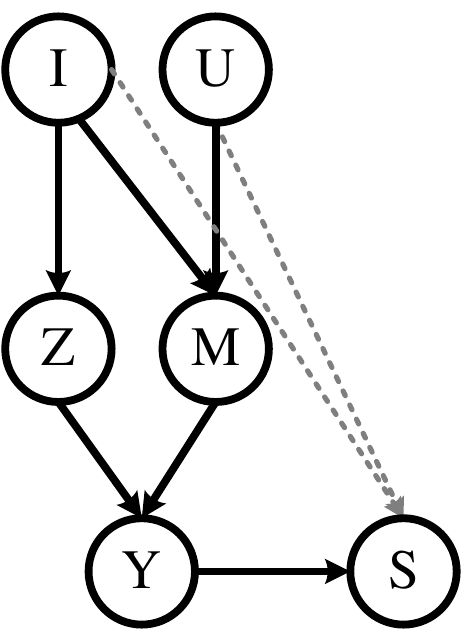}
			%\caption{fig1}
		\end{minipage}%
		\label{causal:causal_1}
	}%
	\subfigure[We cut off $Z\to Y$ during distillation.]{
		\begin{minipage}[t]{0.5\linewidth}
			\centering
			\includegraphics[scale=0.7]{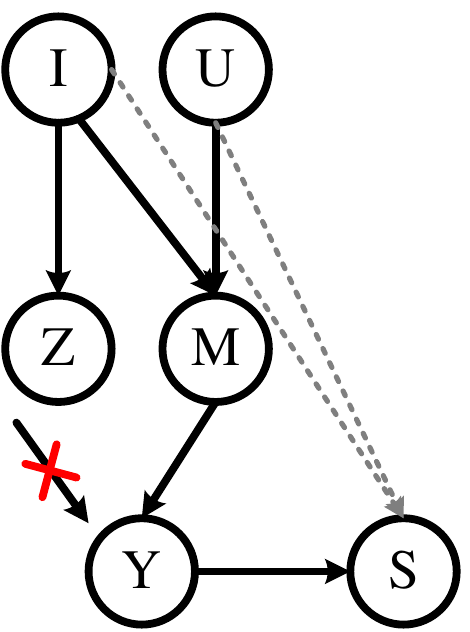}
			%\caption{fig2}
		\end{minipage}%
		\label{causal:causal_2}
	}%

	\centering
	\caption{ The causal graph to describe the knowledge distillation. $U$: user, $I$: item,  $M$ affinity score, $Z$ item popularity, $Y$: soft label, $S$: student. The bias origins from the  causal effect of $Z$ on $Y$. Our UnKD is intended to cut off $Z\to Y$. Admittedly, there may exist other causal paths from $U, I$ to $S$, but here we only focus on the causal effect through distillation (\ie via $Y$).}
\end{figure}%

\begin{figure*}[t]
	\centering
	
	\includegraphics[height=0.27\textheight,width=1\textwidth]{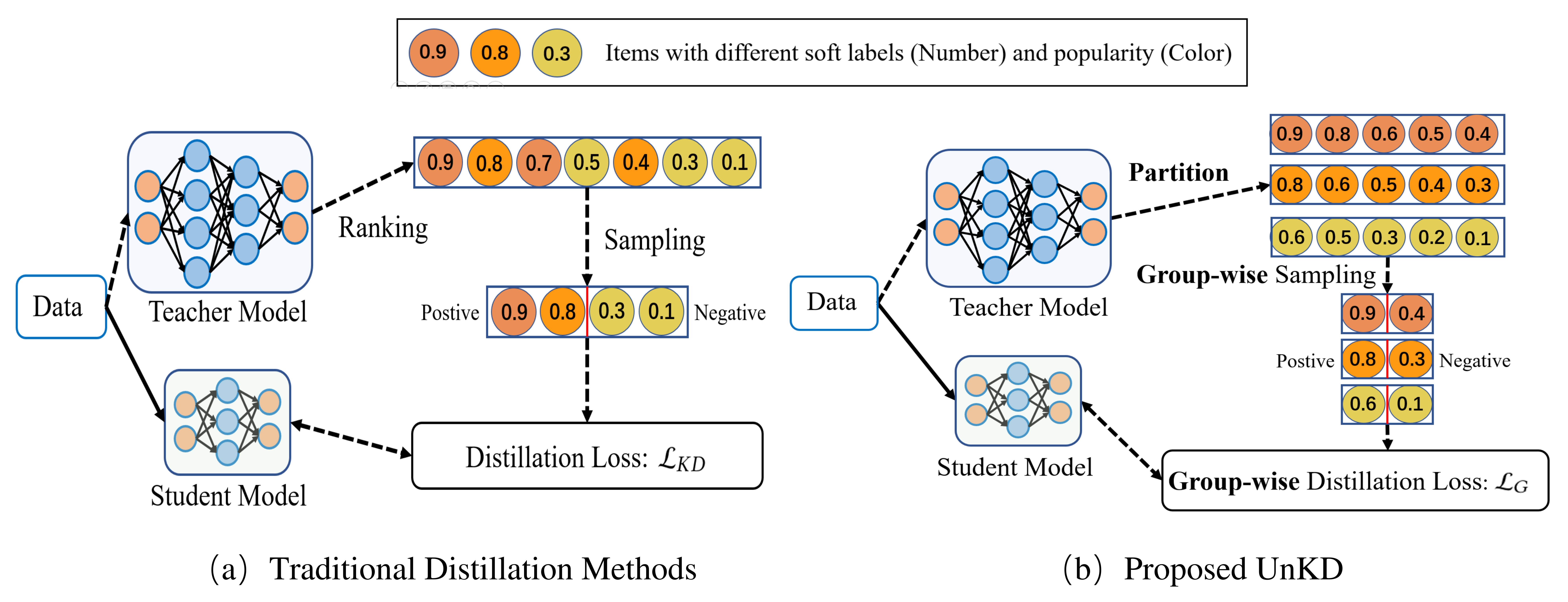}
	\caption{Illustrations of (a) the traditional knowledge distillations and (b) our proposed UnKD. UnKD partitions items into multiple groups according to their popularity, and then extracts the ranking knowledge among each group to learn the student.}
	\label{causal:causal_4}
\end{figure*}%

\section{METHODOLOGY}
In this section, we first resort to a causal graph to trace the origin of bias in knowledge distillation. We then introduce the proposed UnKD and discuss its rationality for eliminating the bias.

\subsection{A Causal View on Distillation Bias}

\textbf{Origin of Distillation Bias.} To trace the origin of distillation bias and to understand how it affects student model, we resort to the language of causal graph \cite{causal:Causality} for a qualitive analysis. Figure \ref{causal:causal_1} illustrates causal relations behind existing distillation methods, which consists of six random variables including:
\begin{itemize} 
	\item $U$ represents a user node, \eg user profile or feature (\eg IDs) that is used for representing a user;
	\item Similar to $U$, Node $I$ represents an item node;
	\item $M$ represents the real affinity score between a user $U$ and an item $I$, reflecting to what extent that the item matches the preference of the user; 
	%\item $G$ represents which group (\eg the genre of movies) that an item belongs to;
	\item $Z$ represents the item popularity;
	\item $Y$ represents the soft label predicted by the teacher model;
	\item $S$ represents the learned student model.
\end{itemize}
The edges in the graph describe the causal relations between variables. Specifically, we have:
\begin{itemize} 
	\item Edges $(U,I)\to M$ depict the causal effect of the features of a user and an item on their affinity;
	\item Edges $I \to Z$ depicts that the popularity of an item is affected by its characteristics;
	\item Edges $(M,Z)\to Y$ show that the soft label $Y$ is affected by two factors: 1) $M\to Y$, the desirable effect from the affinity; 2) $Z \to Y$, the influence from the item popularity, where an item with larger popularity is prone to have higher prediction score. Recent work has validated the effect of $Z \to Y$ is common in recommendation. It can be original from the biased data (\ie popular items is usually over-exposed \cite{10.1145/3404835.3462875}), learning algorithm (momentum-based optimizer is biased towards mainstream \cite{tang2020longtailed}) or recommendation architecture\cite{DBLP:journals/corr/abs-1909-06362} (\ie latent factor models prefer to promote popular groups ).
	\item Edge $Y \to S$ depicts the student model is learned under the supervision of soft labels. 
	%Admittedly, there may exist other causal paths from $U, I$ to $S$. But this work only focuses on the causal effect through distillation (\ie $Y$). Simply omitting other paths do not change our conclusion.
\end{itemize}

According to the causal graph, since there exists an additional path $(I\to Z \to Y)$ from $I$ to $Y$, the learned soft label would be deviated from reflecting user's true preference, \eg an item with higher scores simply because it belongs to a mainstream groups rather than it really meets user preference. Such biases would be further propagated and intensified into the student model, heavily deteriorating its recommendation quality. Typically, the student model would be skewed under the supervision from the biased soft labels; Worse still, note that existing KDs usually employ rank-aware sampling strategy for training a student model. The popular items which usually have abnormally higher scores would obtain more sampling opportunities and thus exert excessive contributions on training. The bias would be amplified during the distillation.
%Since the teacher is trained on a biased dataset, its prediction results and intermediate knowledge are biased. Existing distillation methods directly transfer this knowledge to the student, which will cause additional bias to be introduced into the student's training process. 
As such, it is essential to address bias issue in knowledge distillation. The core lies on blocking the causal effect from $I$ on $Y$ along the path $(I\to Z \to Y)$.

\textbf{Quantifying Bias Effect.} Given the importance of cutting off the path $(I\to Z \to Y)$, here we refer to the language of causal inference \cite{causal:Causality} and give a formula of the causal effect that we aim at estimating. We first quantify the causal effect from the bias and then remove it from the total effect to recover the desirable effect from the preference. 

Let  $Y_{A=a}|U=u$ (short as $Y_{a}|U=u$)  be the random variable with conditional distribution $p(Y|\text{do}(A=a),U=u)$ where a variable $A$ is intervened with a specific value $a$. The causal effect of a variable $I$ on $Y$ is the magnitude by which the target variable $Y$ is changed by a unit change in an variable $I$ \cite{causal:Causality}. For example, the conditional total effect of $I=i$ on $Y$ for a specific user $u$ is defined as:
\begin{equation}
	\text{TE}_i= {Y_i}|u - {Y_{{i^*}}}|u
\end{equation}
which can be understood as the difference between two hypothetical situations $I=i$ and $I=i^*$. $I=i^*$ can be considered as a benchmark situation for comparison. $\text{TE}_i$ can be decomposed into two parts: 1) the desirable causal effect along the path $I\to M \to Y$; and 2) the undesirable causal effect along the path $(I\to Z \to Y)$. By performing different interventions on $I$ along different causal paths, it is possible to isolate the contribution of the causal effect along different paths. 

Specifically, the path-specific causal effect through $(I\to Z \to Y)$ expresses the value change of $Y$ with the item popularity $Z$ change from $Z_i$ to $Z_i^*$:
\begin{equation}
	\text{PEZ}_i= {Y_{{i^*},{Z_i}}}|u - {Y_{{i^*}}}|u
\end{equation}
Accordingly, eliminating the bias can be realized by reducing $\text{PEZ}$ from $\text{TE}$, we have:
\begin{equation}
	\text{PEM}_i=\text{TE}_i-\text{PEZ}_i= {Y_i}|u - {Y_{{i^*},{Z_i}}}|u
\end{equation}
which expresses the value change of $Y$ with changing $i$ to $i^*$ while keeping $Z$ unchanged. This formula blocks the effect along $Z \to Y$ and can be fed into the student model for unbiased distillation.

However, calculating $\text{PEM}$ is difficult, as it involves a counterfactual inference since the item popularity for a specific item $i^*$ is intervened from the factual value $Z_{i^*}$ to the $Z_i$. A naive solution is to directly intervene into the training of the teacher model, \eg employing some debiasing strategies to mitigate the popularity bias \cite{10.1145/3404835.3462875}. However, as discussed before, learning an completely unbiased teacher is usually impractical and unsatisfied. Our empirical studies in Section 4 also validate that this strategy does not bring satisfactory results. 
%which however is usually impractical and unsatisfied. A large model is usually complicated and adjusting its training procedure would be difficult and expensive.
Therefore, a new unbiased distillation method without requiring to intervene teacher model deserves exploration.

\subsection{Unbiased Knowledge Distillation}
Towards this end, in this work we propose an unbiased knowledge distillation strategy (UnKD), which conducts debiasing during the learning of the student model. The subtlety of UnKD lies on its popularity-stratified training strategy, where the unfeasible counterfactual terms have been properly offset. To be more specific, UnKD first partitions items into multiple groups according to the item popularity, where the items in one group have similar popularity. After that, for each user, UnKD ranks the items on the same group \wrt the soft label, and transfers such group-wise knowledge to supervise the learning of the student model. In fact, we have the following lemma:
\newtheorem{lem}{Lemma}
\begin{lem}
	\label{la1}
	For each user $u$, for the items with highly similar popularity, the list ranked by $Y_i$ is approximately equal to the list ranked by $\text{PEM}_i$.
\end{lem}
\begin{proof}
	For arbitrary two items $i$ and $j$ with highly similar popularity, we have $Z_i \approx Z_j$ and thus the equation ${Y_{{i^*},{Z_i}}}|u={Y_{{i^*},{Z_j}}}|u$ almost holds. Then, we have:
	\begin{align}
		{Y_i}|u > {Y_j}|u &\Leftrightarrow {Y_i}|u - {Y_{{i^*},{Z_i}}}|u > {Y_j}|u - {Y_{{i^*},{Z_j}}}|u \notag \\
		&\Leftrightarrow \text{PEM}_i > \text{PEM}_j
	\end{align}
	The lemma gets proofed.
\end{proof}
It means that the group-wise ranking lists are approximately unbiased, which provide more accurate evidence on users' true preference. UnKD extracts such accurate popularity-stratified ranking knowledge for training a student model, which avoids disturbance from the terrible popularity effect.      

\textbf{Details of UnKD.} The detailed training procedure of UnKD is illustrated in Figure \ref{causal:causal_4}(b). UnKD follows the recent advanced strategy CD \cite{10.1145/3340531.3412005}, differing in employing group-wise sampling and training. UnKD consists of the following three steps:  

(1) \textit{Group partition.} We partition items into $K$ groups according to the item popularity. The partition procedure refers to the recent work \cite{10.1145/3404835.3462875}. Specifically, we first sort items according to their popularity in descending order, and then divide the items into $K$ groups. The items with similar popularity are positioned into the same group. Also, we follow  \cite{10.1145/3404835.3462875} and let the sum of popularity over items in each group is the same. 

We remark that $K$ is an important hyperparameter balancing the trade-off between the unbiasedness and informativeness. A larger $K$ suggests a more fine-grained partition and the items in each group would have higher similarity on popularity. It means the unbiasedness is more likely to be held. However, larger $K$ would decrease the number of items in each group, and reduced the knowledge about the item ranking relations. On the contrary, a smaller $K$ could bring more information but at the expense of unbiasedness. The empirical results of how $K$ affects distillation performance are shown in Section 4.      

%Items are placed into $K$ different groups in two steps: 1) we sort items according to their  popularity in descending order and 2) we divide the items into $K$ groups and ensure that the sum of popularity over items in each group is the same. As such, the number of items will increase from group 1 to group $K$ and the average  popularity among items in each group decreases from group 1 to group $K$. Hence we say group 1 is the most popular group and group $K$ is the least popular group. 

(2) \textit{Group-wise Sampling.} For each user, we first rank the items in each group in terms of the soft labels from the teacher. We then sample a set $\mathcal S_{ug}$ of positive-negative item pairs ($i^+,i^-$) for each group $g$ with the rank-aware probability distribution: $\emph{p}_{i}\propto e^{-rank(i)/\mu}$ \cite{10.1145/3340531.3412005}. Here $rank(i)$ represents the ranking position of $i$ in the group, and $\mu$ is a hyperparameter.

(3) \textit{Group-wise Learning.} We adopted group-wise distillation loss for training a student model:
\begin{align}
	\mathcal{L}_{G}=-\sum_u{\frac{1}{|\mathcal{U}|}\sum_{g\in \mathcal G}\sum_{(i^+,i^-)\in \mathcal{S}_{ug}} \log\sigma(\textbf{e}_{u}^T \textbf{e}_{i^{+}}-\textbf{e}_{u}^T \textbf{e}_{i^{-}})}	
\end{align}
Here the item pairs utilized for loss calculation are sampled from the stage (1) ($(i^+,i^-)\in \mathcal{S}_{ug}$). Note that the item pairs are on the same group and its relations are consistent with user true preference. The distillation loss would be accurate and provide additional useful knowledge for training a better student model. Here, $\textbf{e}_u$ and $\textbf{e}_i$ represent the embedding of $u$ and $i$, respectively. And $\sigma$ represents sigmoid function.

The final loss function for training a student model is:
\begin{align}
	\mathcal{L} =\mathcal{L}_{R} + \lambda \mathcal{L}_{G}
	\label{equ:Equation_7}
\end{align}
where the distillation loss $\mathcal{L}_{G}$ is usually accompanied with the original supervised loss $\mathcal{L}_{R}$ from the training data. Hyperparameter $\lambda$ is utilized to balance their contributions.

% \textbf{Towards fairer distillation.}	
% One may argue that unbalance distribution of items may direct affect student model learning. To tackle this concern, we improve UnKD to UnKD-F with the following strategy:

% Our formula for calculating the number of soft-labels obtained by each group is shown in Equation (\ref{equ:Equation_5}). For each user, the less popular groups will get more soft-labels ($|\mathcal {S}_{ug}|$) if the total number of soft-labels remains the same ( $|\mathcal {S}_{u}|$). Items in the unpopular group will be given more training opportunities to mitigate the effects of popularity bias.

% \begin{align}
	% 	D_{g}=(Max(\mathcal Z)&+Min(\mathcal Z))-Z_{g}, \nonumber\\
	% 	P_{g}=D_{g}&/\sum_{j}^{\mathcal G}D_{j},	\nonumber\\
	% 	|\mathcal {S}_{ug}|=&|\mathcal {S}_{u}|*P_{g}
	% 	\label{equ:Equation_5}
	% \end{align}%
%  There may be better allocation schemes, which will be our future work. For the UnKD method, we set the number of soft-labels per for each group to $ |\mathcal {S}_{u}|/K$. 

\section{EXPERIMENTS}
In this section, we conduct experiments to evaluate the performance of our proposed UnKD. Our experiments are intended to address the following research questions: 

\begin{itemize}
	\item [\textbf{RQ1:}] How does UnKD perform compared with existing distillation methods? Does UnKD benefit unpopular items? 
	\item [\textbf{RQ2:}] Does UnKD outperform the strong baseline that leverages the debiasing techniques in teacher model training? 
	\item [\textbf{RQ3:}] How does the hyperparameter $K$ (Group numbers) affect distillation performance?
\end{itemize}

%	RQ3:  What has changed in UnKD-F compared to UnKD? 

\subsection{Experimental Setup}

\textbf{Datasets.} Three commonly-used datasets are adopted for testing the model performance including \textbf{CiteULike}\footnote{https://github.com/changun/CollMetric}, \textbf{Amazon-Apps}\footnote{http://jmcauley.ucsd.edu/data/amazon/links.html}, and \textbf{MovieLens-1M}\footnote{https://grouplens.org/datasets/movielens/}. 
For stable evaluation, we filter out users with fewer than 20 interactions. Also we transform the detailed rating value into binary for implicit recommendation as recent work \cite{10.1145/3447548.3467319}. The statistics of the datasets are shown in Table \ref{tab:table_1}.
Besides, for each user we randomly select 90\% of historical interactions as the training set, and the remaining 10\% data constitutes the testing set. We also randomly partition 10\% interactions from training data for model validation.

\begin{table}[t] %
	\centering
	\caption{Statistics of the  datasets.}
	\begin{tabular}{|c|c|c|c|c|}
		\hline
		dataset&Users&Items&Interactions&Sparsity \\
		\hline
		CiteULike   & 5219  & 25181& 125580 &99.91\% \\
		Apps       & 3898  & 11797& 128105 &99.73\% \\
		MovieLens    & 6040  & 3706&1000209&95.54\% \\
		\hline
	\end{tabular}
	
	\label{tab:table_1}
\end{table}%

\textbf{Compared Methods.} We compare our methods with the following baselines:
\begin{itemize}
	\item RD \cite{10.1145/3219819.3220021}: A classic KD method for recommendation that treats the Top-N ranked items as positive while reweighs the items according to the position.
	\item CD \cite{8970837}: A method that creates positive-negative training instances based on the item ranking position from the teacher. 
	\item DERRD \cite{10.1145/3340531.3412005}: A KD that trains a student model from both teachers' prediction and teacher latent knowledge.
	%They sample items from teachers’ predictions based on their ranking, then use them for distillation.
	\item HTD \cite{10.1145/3447548.3467319}: An advanced KD method that distills the topological knowledge from the teacher embedding space.
\end{itemize}

%The base model for all distillation methods are BPR-MF \cite{DBLP:journals/corr/abs-1205-2618} and LightGCN\cite{10.1145/3397271.3401063}. 
We test the above distillation methods on two representative backbone models: BPR-MF \cite{DBLP:journals/corr/abs-1205-2618} and LightGCN\cite{10.1145/3397271.3401063}. We also report the performance of teacher and student models that are directly trained from the training dataset.

%We also report two versions of our methods: the basic version FairKD and a fairer version FairKD-F, which balances the popularity in the data via fair sampling.

\begin{table}[t]\small%
	\centering
	\caption{Overall performance comparison between our method and  baselines. All metrics are based on the top-10 results. where the best performance is bold and the second best underlined.}
	\begin{tabular}{|c|c|c|c|c|c|}
		\hline
		&Backbone Model&\multicolumn{2}{c|}{BPRMF}&\multicolumn{2}{c|}{LightGCN}\\
		\hline 
		Dataset&Method&Recall&NDCG&Recall&NDCG\\
		\hline
		
		\multirow{8}{*}{MovieLens}&Teacher&0.1810&0.2951&0.1850&0.3012\\
		\cline{2-6}
		&Student&0.1435&0.2511&0.1456&0.2581\\
		&RD&\underline{0.1473}&\underline{0.2559}&0.1471&0.2583\\
		&CD&0.1445&0.2534&0.1477&0.2602\\
		&DERRD&0.1436&0.2532&\underline{0.1487}&\underline{0.2606}\\
		&HTD&0.1441&0.2539&0.1472&0.2592\\
		\cline{2-6}
		%&UnKD-F&0.1933&0.2377&0.2072&0.2586\\
		&UnKD&\textbf{0.1547}&\textbf{0.2615}&\textbf{0.1569}&\textbf{0.2672}\\
		&impv-e\%&5.02\%&2.18\%&5.51\%&2.53\%\\
		\hline
		\multirow{8}{*}{Apps}&Teacher&0.0991&0.0760&0.1007&0.0782\\
		\cline{2-6}
		&Student&0.0719&0.0539&0.0811&0.0643\\
		&RD&0.0768&0.0596&0.0831&0.0647\\
		&CD&\underline{0.0790}&\underline{0.0608}&\underline{0.0848}&\underline{0.0658}\\
		&DERRD&0.0729&0.0562&0.0832&0.0648\\
		&HTD&0.0732&0.0561&0.0833&0.0652\\
		\cline{2-6}
		%&UnKD-F&0.0272&0.0587&0.0295&0.0651\\
		&UnKD&\textbf{0.0853}&\textbf{0.0644}&\textbf{0.0867}&\textbf{0.0678}\\
		&impv-e\%&7.97\%&5.92\%&2.24\%&3.04\%\\
		\hline
		\multirow{8}{*}{CiteULike}&Teacher&0.1518&0.1016&0.1657&0.1139\\
		\cline{2-6}
		&Student&0.0760&0.0477&0.0783&0.0510\\
		&RD&\underline{0.0808}&0.0514&0.0833&0.0538\\
		&CD&0.0801&\underline{0.0518}&0.0936&0.0616\\
		&DERRD&0.0793&0.0511&0.0809&0.0527\\
		&HTD&0.0788&0.0485&\underline{0.0958}&\underline{0.0628}\\
		\cline{2-6}
		%&UnKD-F&0.0172&0.0501&\underline{0.0212}&0.0617\\
		&UnKD&\textbf{0.0863}&\textbf{0.0550}&\textbf{0.1006}&\textbf{0.0654}\\
		&impv-e\%&6.80\%&6.17\%&5.01\%&4.14\%\\
		\hline
	\end{tabular}
	\label{tab:table_2}
\end{table}%

\textbf{Implementation Details.} For the backbone model, we closely refer to \cite{10.1145/3340531.3412005} and set the embedding dimension of the teacher as 100 and the student as 10. Adam is adopted as our optimizer. The search space of the learning rate for all experiments is \{0.01,0.001,0.0001\}, and the space of the L2 regularization coefficient is \{0.01,0.001,0.0001\}. We adopt the early stopping strategy that stops training if $NDCG$ on the validation data does not increase for 100 epochs. The total number of training epochs is set to 1000 epochs. For the compared baselines, we closely follow their settings reported in the relevant papers or directly utilize their codes if they are available. We also finely tuned their hyperparameters to ensure optimum.

For our method, during the training phase, the number of groups $K$ is set in the range $\{2,3,4,...,10\}$. In the testing phase, for better visualization, we only divide items into two groups, popular group and unpopular group.  $\lambda$ is set in the range $\{0.1,0.2,0.3,...,1.0\}$, and $\mu$ is set in the range $\{10,20\}$. For each user, the number of soft-labels is set in the range \{30, 40\}.

\textbf{Evaluation Metrics.} The conventional ranking metrics including normalized
discounted cumulative gain (NDCG@N), and Recall (Recall@N) are adopted for evaluating model performance. We also report Recall for the popular (or unpopular) group, \ie estimating the fraction of relevant popular (or unpopular) items that are in the top-N ranking list. This metric can reflect how well the model retrievals the popular (or unpopular) items. In this work, we simply choose N as 10. 

\subsection{Performance Comparison (RQ1)}

\textbf{Overall performance comparison.}  Table \ref{tab:table_2} shows the overall performance of our UnKD compared with other KD methods. We observe our UnKD consistently outperforms other KDs on all three datasets. Especially in the dataset CiteULike, the improvements are encouraging. UnKD achieves 5.53\% on average improvement over the baselines. Obviously, this result validates that addressing bias issue in knowledge distillation is essential and indeed boosts distillation performance. 

\textbf{Comparison in terms of popular/unpopular groups.} To understand how our UnKD addresses bias issue in knowledge distillation, we also report the performance (recall@10) for popular and unpopular item groups. As the results for popular/unpopular groups may have different scale, here we report the relative improvements over the student baseline for better presentation. Figure \ref{fig:image_3} illustrates the results. We make the following observations: 1) the improvements of existing knowledge distillation methods are mainly from the popular items, while the performance of unpopular items severely suffers. 2) The improvements of UnKD mainly lies on unpopular items. Especially in the dataset CiteULike, UnKD achieves over 100\% performance gain for unpopular items. Our UnKD could indeed address bias issue in knowledge distillation, yielding more accurate and fair recommendations.

\subsection{Distillation Procedure vs. Teacher Training (RQ2)}

\begin{table*}[ht]%
	\centering
	\caption{Performance comparison (recall@10) between our UnKD and the baselines that leverages debiasing technique in model training. The best performance is shown in bold, and the second best performance is underlined.}
	\begin{tabular}{|c|c|c|c|c|c|c|c|c|c|c|}
		\hline
		&Dataset&\multicolumn{3}{c|}{MovieLens}&\multicolumn{3}{c|}{Apps}&\multicolumn{3}{c|}{CiteULike}\\
		\hline 
		\multirow{2}{*}{Backbone Model}&\multirow{2}{*}{Method}&\multirow{2}{*}{Overall}&Popular&Unpopular&\multirow{2}{*}{Overall}&Popular&Unpopular&\multirow{2}{*}{Overall}&Popular&Unpopular\\
		&&&Group&Group&&Group&Group&&Group&Group\\
		\hline
		& Student  &0.1435&0.2156&0.0250& 0.0719& 0.1031&0.0109 &0.0760&0.0831&\underline{0.0095}\\
		& CD  & 0.1445&\textbf{0.2258}&0.0179  & 0.0790&0.1212&0.0090 &0.0801&0.0887&0.0088 \\
		& PD-CD  &\underline{0.1454}&0.2205&0.0210&0.0795&0.1176&\underline{0.0113} &\underline{0.0805}&\underline{0.0890}&0.0092 \\
		%		& Reg-CD  &0.1446& \underline{0.2245}&0.0177& 0.0785&0.1173&0.0101&0.0762&0.0844&0.0063  \\
		BPRMF&HTD&0.1441&\underline{0.2228}&0.0187 &0.0732&0.1195&0.0061&0.0788&0.0885&0.0068\\
		& PD-HTD  &0.1443&0.2150&\underline{ 0.0263} & \underline{0.0808}&\underline{0.1240}&0.0076 & 0.0798&\textbf{0.0897}&0.0073\\
		%		& Reg-HTD  & 0.1441& 0.2204&0.0206 & 0.0759&0.1175&0.0079&0.0765&0.0858&0.0050  \\
		\cline{2-11}
		& UnKD  &\textbf{0.1547}&0.2205 &\textbf{0.0311}&\textbf{0.0853}&\textbf{0.1274}&\textbf{0.0147}&\textbf{0.0863}&0.0854 &\textbf{0.0208}\\
		\hline
		& Student &0.1456&0.2280&\underline{0.0228}  &0.0811&0.1242&0.0093 &0.0783&0.0885&0.0080 \\
		& CD  &0.1477&0.2316&0.0169&0.0848&\underline{0.1310}&0.0091 &0.0936&0.1067&0.0081 \\
		& PD-CD  &\underline{0.1496}&\underline{0.2369}&0.0172&\underline{0.0851}&0.1308&\underline{0.0095}&0.0942&0.1020&0.0110 \\
		%		& Reg-CD    &\underline{0.1502}&\textbf{0.2426}&0.0122&0.0856&0.1299&0.0096 &0.0851&0.0928& 0.0061\\
		LightGCN&HTD&0.1472&0.2328&0.0079&0.0833&0.1291&0.0091&0.0958&\underline{0.1093}&0.0085 \\
		& PD-HTD  &0.1485&0.2364&0.0159 &0.0835&0.1288& 0.0093&\underline{0.0979}&\textbf{0.1102}&\underline{0.0128}\\		
		%		& Reg-HTD  &0.1475&0.2345&0.0172 &\underline{0.0860}&\underline{0.1314}&0.0085  &0.0932&\underline{0.1082}&0.0075\\
		\cline{2-11}
		& UnKD  &\textbf{0.1569}&\textbf{0.2384}&\textbf{0.0292}&\textbf{0.0867}&\textbf{0.1325}&\textbf{0.0118}&\textbf{0.1006}&0.1076&\textbf{0.0162}\\
		
		\hline			
	\end{tabular}
	
	\label{tab:table_4}
\end{table*}%

%	\begin{figure*}[t]
	%		\centering
	%		
	%		
	%		\includegraphics[scale=0.45]{Figure_ratio_group.pdf}
	%		
	%		\caption{On the three datasets, with the change of $\beta$, the change trend of $NDCG$ and $DTR$. }
	%		\label{fig:image_9}
	%	\end{figure*}%
\begin{figure*}[t]
	\centering
	
	\includegraphics[scale=0.45]{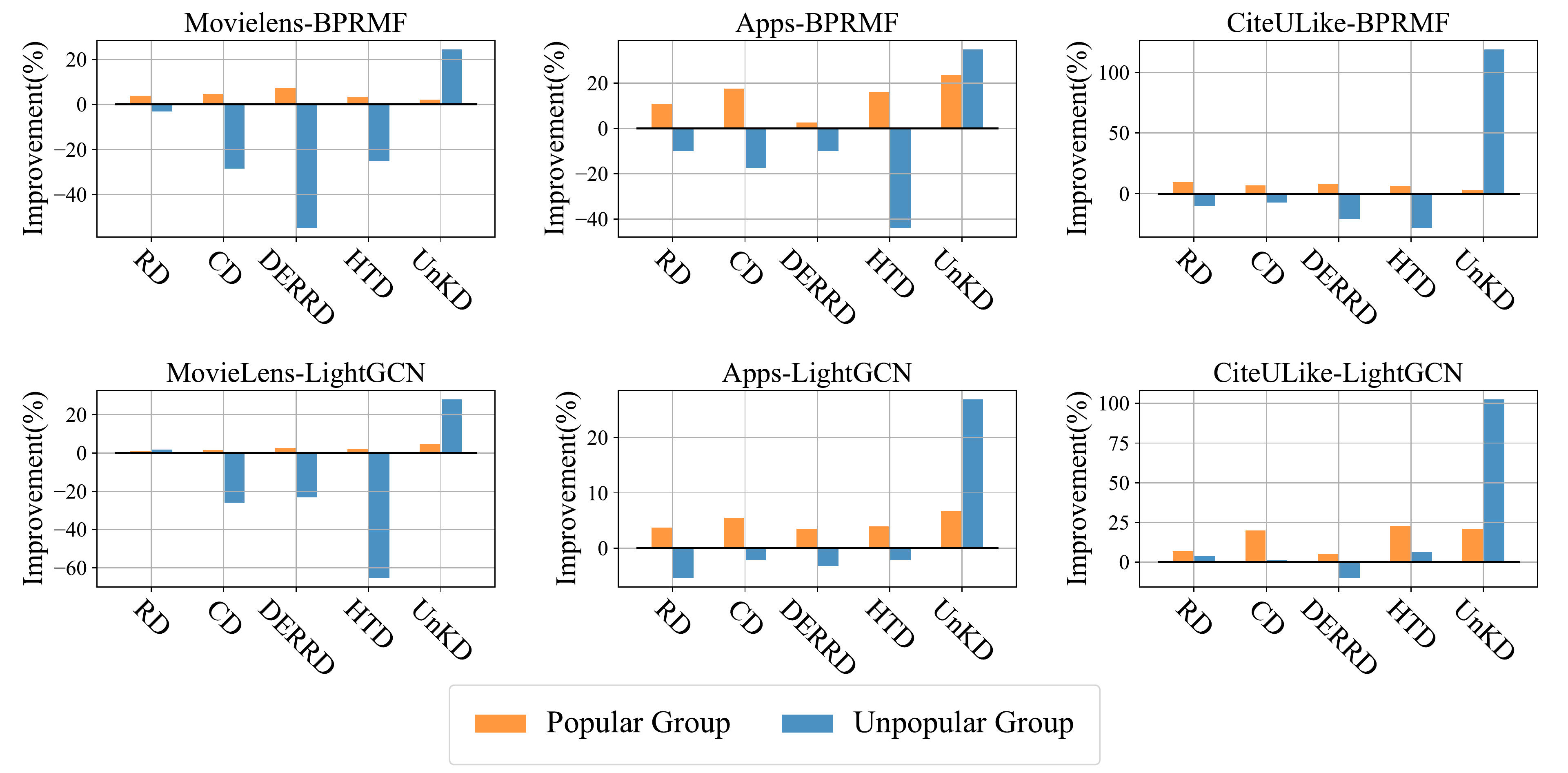}
	\caption{The relative improvements (\wrt recall@10) of KDs  over the baseline that directly trained from the dataset. Here we visualize the results in terms of popular and unpopular group, respectively.}
	\label{fig:image_3}
\end{figure*}%

\begin{figure*}[th]
	\centering
	
	\includegraphics[scale=0.45]{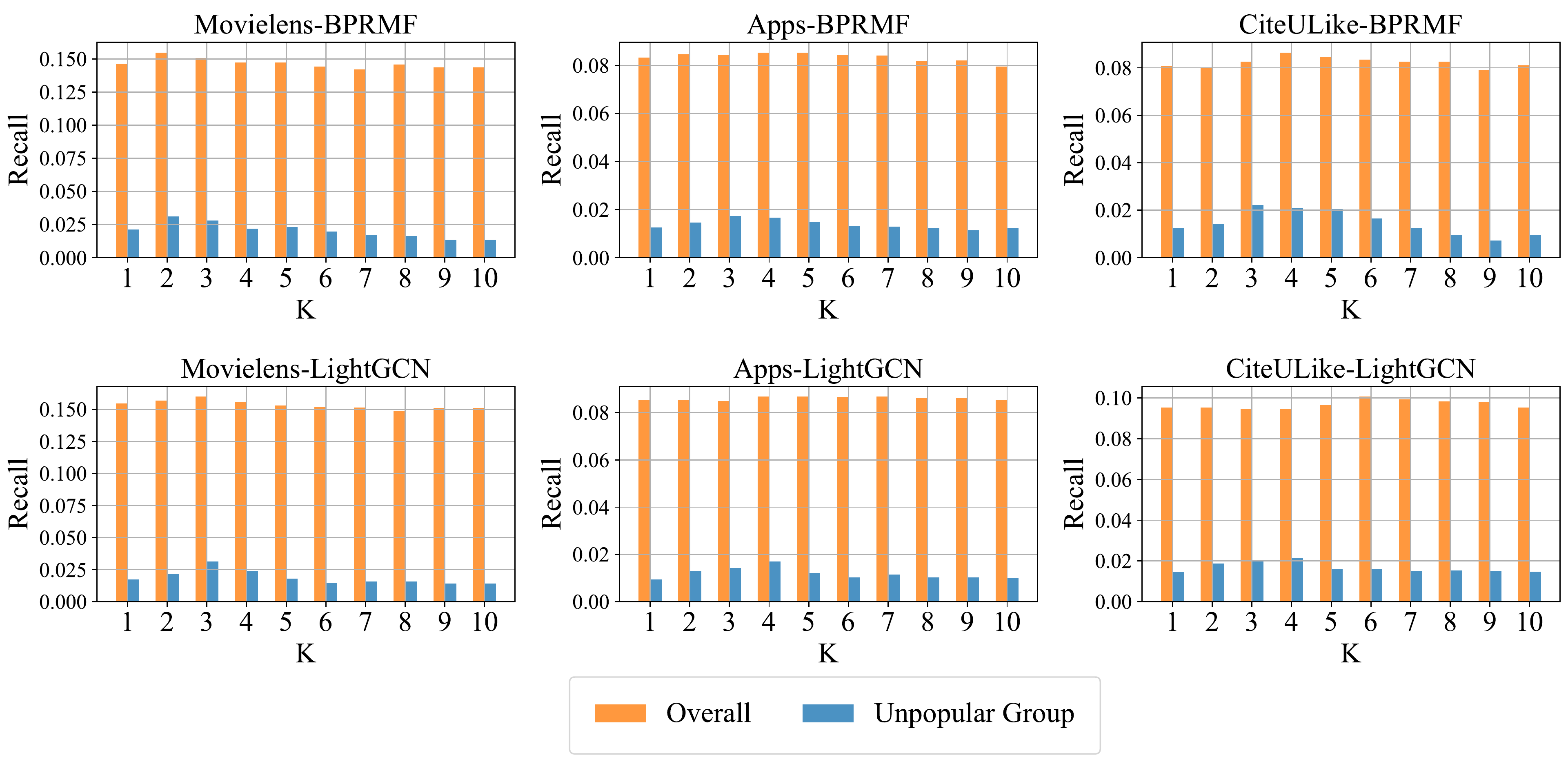}
	\caption{Performance comparison with varying $K$.}
	\label{fig:image_5}
\end{figure*}%

Although previously we have discussed that directly intervening the teacher model training for debiasing is not a good choice, we are still curious about its performance. Here we compare UnKD with a strong baseline that leverages an advanced debiasing technology (PD \cite{10.1145/3404835.3462875}) in teacher model training. PD leverages causal inference to tackle the popularity bias, and usually achieves state-of-the-art performance in a widely range of datasets. We integrate PD into two SOTA KDs (\ie PD-CD, PD-HTD) for comparisons. 

Table \ref{tab:table_4} presents the results. We make the following observations: Leveraging PD in teacher model training could boost the performance of unpopular items. However, the improvements are not significant as our UnKD. The reason can be attributed to the complexity of the bias in teacher. The bias may roots in many factors. Existing debiasing methods are usually tailored for one or two specific factors and may not eliminate the bias accurately and completely. Also, an improper debiasing may hurt model accuracy. UnKD could circumvent this challenging problem and does not require to intervene the training of the cumbersome teacher model, which is more effective and satisfactory.

\subsection{Effect of the Parameter $K$ (RQ3)}
It will be interesting to explore how hyper-parameter $K$ affects the performance of UnKD, where $K$ indicates the number of partition groups in the distillation. Figure \ref{fig:image_5} illustrates the results (Recall@10) on all items and unpopular items, respectively. 

As can be seen, with the number of groups ($K$) increasing, with few exception, the performance on unpopular items will become better first. The reason is that a larger $K$ suggests a more fine-grained partition and the items in each group is prone to have higher popularity similarity. The unbiasedness of the distillation is more likely to be held. However,  when $K$ surpasses a threshold, the performance becomes worse with further increase of $K$. The reason is that a further larger $K$ would make the number of items in each group decrease. The knowledge about some item ranking relations is missing. As such, $K$ balances the trade-off between the informativeness and unbiasedness. Set $K$ to a proper value (\eg $K=4$) could achieve best performance for unpopular items. Similar results are observed for the overall performance, except it is relatively stable. The performance of popular items is relatively robust to $K$. This is because popular items can also benefit from the rich label information from the data. 

\section{Related Work}
In this section, we review the most related work from the following three perspectives.

\textbf{Knowledge Distillation in Recommendation.}  Knowledge distillation (KD) is a promising model compression technique that first trains a large teacher model and then transfers the knowledge from the teacher to the target compact student model \cite{10.1145/3219819.3220021,Yim_2017_CVPR,DBLP:journals/corr/abs-1710-09282,zhang2020distilling}. KDs have been widely applied in recommender systems to reduce inference latency. They mainly utilize soft labels (\ie teacher predictions) for knowledge transfer. For example, RD \cite{10.1145/3219819.3220021} ranked the soft labels from the teacher and treated the top-N ranked items as positive for training a student model; CD \cite{8970837} utilized soft labels to create positive and negative distillation instances; Soft labels also have been considered by DERRD \cite{10.1145/3340531.3412005} to create the list-wise distillation loss function. In additional to soft labels, some work considered to transfer the hidden knowledge among the middle layer of teachers (\eg latent representation).  For example, DERRD \cite{10.1145/3340531.3412005} leveraged expert neural networks to extract useful information from the teacher representations; HTD \cite{10.1145/3447548.3467319} distilled the topological knowledge built upon the relations in the teacher embedding space. Despite their decent performance, we remark that existing distillation methods are severely biased towards popular items. 

Besides model compression, there are also some other applications of KDs in recommender system \cite{10.1145/3397271.3401083,2022Cross,xia2022device,10.1145/3459637.3482117,ding2022interpolative}. For example, in the social recommendation, KD is used to integrate the knowledge from various relational graphs \cite{10.1145/3485447.3512003}; KD also plays an important role in tackling data selection bias \cite{DBLP:journals/corr/abs-1910-01444,DBLP:journals/corr/abs-2201-08024}; Some work also considers to leverage KD for model ensemble\cite{zhu2020ensembled}.

\textbf{Bias in Recommendation.} As this work focuses on popularity bias, here we mainly review recent work on this bias. For other types of biases and their debiasing techniques, we simply refer readers to a comprehensive survey \cite{DBLP:journals/corr/abs-2010-03240} for more information. 

Popularity bias depicts a common phenomenon \cite{DBLP:journals/corr/abs-2010-03240} that Popular items are recommended even more frequently than their popularity would warrant. Ignoring the popularity bias will result in many severe issues like affecting recommendation accuracy, decreasing recommendation diversity, and even raising ``Matthew effect''.
To tackle popularity bias, recent work mainly lies on three types: 1) leveraging suitable regularization in model learning or ranking to push the model towards balanced recommendation lists \cite{10.1145/3437963.3441820};  2) conducting adversarial training to improve the recommendation opportunity of the niche items \cite{10.1145/3269206.3269264}; 3) resorting to causal graph to identify the origin of the bias and conduct debiasing accordingly \cite{10.1145/3404835.3462875}. Although existing methods on popularity bias have achieved great progress, how to completely eliminate popularity bias is still an open problem. Popularity bias is seriously complicated and may root in various components including but not limited to optimizer \cite{tang2020longtailed}, model architecture \cite{DBLP:journals/corr/abs-1909-06362}, or training data \cite{10.1145/3404835.3462919}. In RS, popularity bias is also occurred during the knowledge distillation, which has not been explored.

%There are two main goals for the fair recommendation. One is to make each recommended group be treated fairly, which is enforced fairness. The other is to prevent the increase of bias and lead to unfair recommendations. Its goal is to keep the proportion of each group in the original data set as consistent as possible with the proportion of the recommended results \cite{10.1145/3366424.3380048,tang2020longtailed,10.1145/3442381.3450015,DBLP:journals/corr/abs-1907-13286,10.1145/3437963.3441820}. The goal of fairness recommendation is to make the exposure of different groups similar in recommendation results \cite{DBLP:journals/corr/abs-2104-05994,10.1145/3442381.3449866}. This is often a trade-off issue, most of the existing methods will gain fairness by sacrificing certain accuracy \cite{Biega2018EquityOA,10.1145/3437963.3441820}. We believe this is because they did not improve accuracy on certain groups when they increased their exposure. Compared with existing methods, we not only consider the exposure rate of each group, but also the accuracy rate of each group, which allows us to achieve fair recommendations while also improving the accuracy of the model. Unlike most methods that use re-ranking techniques to achieve fairness, we implement an end-to-end method\cite{10.1145/3404835.3462966,DBLP:journals/corr/abs-1909-06362,NEURIPS2019_44a2e080,10.1145/3442381.3450015}.

\textbf{Causal Recommendation.} Causal inference has received increasing attention in the field of machine learning \cite{Niu_2021_CVPR,10.1145/3404835.3462971,10.1145/3240323.3240360}. In recommender systems, causal inference can be utilized for tackling bias \cite{10.1145/3404835.3462875}, making explainable recommendation \cite{xu2021learning} or improving model generalization. As this work focuses on bias, here we mainly review recent work on causality-enhanced debiasing. They can be classified into three types: 1) The most well-known causal strategy for debiasing is IPS, which reweighs instances with the inverse of the propensity scores. IPS has been widely for tackling various bias, including position bias \cite{10.1145/3331184.3331202}, selection bias \cite{pmlr-v48-schnabel16}, and exposure bias \cite{10.1145/3240323.3240355}. 2) Another type of causality-based debiasing would resort to a causal graph. They leverage the causal graph to trace the origin of bias, and then perform counterfactual inference to cut off the effect from the bias such as PDA \cite{10.1145/3404835.3462875}, MACR \cite{10.1145/3447548.3467289}. 
3) The last relies on constructing counterfactual instances \cite{10.1145/3459637.3482305}. This method uses counterfactual inference to produce counterfactual instances that are used to offset the bias.

%most of the work focussing dases on removing biases in user feedback data  \cite{DBLP:journals/corr/abs-2010-03240,10.1145/3460231.3474263}, like selection bias, position bias, etc. One of the most mainstream methods is based on IPS($Inverse$ $Propensity$ $Scoring$) \cite{pmlr-v48-schnabel16,10.1145/3404835.3463118}. However, since this method is sensitive to the calculation of inverse propensity score and difficult to generalize to different application scenarios, the causal intervention has recently received more and more attention \cite{10.1145/3404835.3462875,10.1145/3447548.3467289}. A part of the work has recently emerged using causal inference to implement the fairness issue of recommendations \cite{10.1145/3404835.3462966,https://doi.org/10.48550/arxiv.1703.06856}. To our knowledge, in the recommendation domain, we are the first to apply causal inference to the knowledge distillation process.

\section{Conclusion}
In this work, we studies on an important but unexplored problem --- bias issue in distilling a recommendation model. we first identify the origin of the bias --- it roots in the biased soft labels from the teacher, and is further propagated and intensified during the distillation. To rectify this, we proposes an unbiased teacher-agonistic knowledge distillation (UnKD) that extracts popularity-aware ranking knowledge to guide student learning. Our experiments on three real-world datasets validate that our UnKD outperforms state-of-the-arts by a large margin, especially for unpopular item group.

Note that this work only explores distillation bias from the popularity perspective. One interesting direction for future work is to explore more fine-grained bias (\eg feature-level fairness) in knowledge distillation.  
Also, considering sequential recommendation is drawing increasingly attention, it will be valuable to explore the model compression technique for the large sequential recommendation models. 

\section{ACKNOWLEDGMENTS}
This work is supported by the National Key Research and Development Program of China (2021ZD0111802), the National Natural Science Foundation of China (62102382, 62272437, 62106221), the CCCD Key Lab of Ministry of Culture and Tourism and the Starry Night Science Fund of Zhejiang University Shanghai Institute for Advanced Study (SN-ZJU-SIAS-001).

\bibliographystyle{ACM-Reference-Format}
\bibliography{sample-manuscript}


\begin{thebibliography}{42}

%%% the Web address.

\ifx \showCODEN    \undefined \def \showCODEN     #1{\unskip}     \fi
\ifx \showDOI      \undefined \def \showDOI       #1{#1}\fi
\ifx \showISBNx    \undefined \def \showISBNx     #1{\unskip}     \fi
\ifx \showISBNxiii \undefined \def \showISBNxiii  #1{\unskip}     \fi
\ifx \showISSN     \undefined \def \showISSN      #1{\unskip}     \fi
\ifx \showLCCN     \undefined \def \showLCCN      #1{\unskip}     \fi
\ifx \shownote     \undefined \def \shownote      #1{#1}          \fi
\ifx \showarticletitle \undefined \def \showarticletitle #1{#1}   \fi
\ifx \showURL      \undefined \def \showURL       {\relax}        \fi
% invisible to TeX
\providecommand\bibfield[2]{#2}
\providecommand\bibinfo[2]{#2}
\providecommand\natexlab[1]{#1}
\providecommand\showeprint[2][]{arXiv:#2}

\bibitem[Agarwal et~al\mbox{.}(2019)]%
        {10.1145/3331184.3331202}
\bibfield{author}{\bibinfo{person}{Aman Agarwal}, \bibinfo{person}{Kenta
  Takatsu}, \bibinfo{person}{Ivan Zaitsev}, {and} \bibinfo{person}{Thorsten
  Joachims}.} \bibinfo{year}{2019}\natexlab{}.
\newblock \showarticletitle{A General Framework for Counterfactual
  Learning-to-Rank}. In \bibinfo{booktitle}{\emph{Proceedings of the 42nd
  International ACM SIGIR Conference on Research and Development in Information
  Retrieval}} (Paris, France) \emph{(\bibinfo{series}{SIGIR'19})}.
  \bibinfo{publisher}{Association for Computing Machinery},
  \bibinfo{address}{New York, NY, USA}, \bibinfo{pages}{5–14}.
\newblock
\showISBNx{9781450361729}
\urldef\tempurl%
\url{https://doi.org/10.1145/3331184.3331202}
\showDOI{\tempurl}


\bibitem[Bonner and Vasile(2018)]%
        {10.1145/3240323.3240360}
\bibfield{author}{\bibinfo{person}{Stephen Bonner} {and}
  \bibinfo{person}{Flavian Vasile}.} \bibinfo{year}{2018}\natexlab{}.
\newblock \showarticletitle{Causal Embeddings for Recommendation}. In
  \bibinfo{booktitle}{\emph{Proceedings of the 12th ACM Conference on
  Recommender Systems}} (Vancouver, British Columbia, Canada)
  \emph{(\bibinfo{series}{RecSys '18})}. \bibinfo{publisher}{Association for
  Computing Machinery}, \bibinfo{address}{New York, NY, USA},
  \bibinfo{pages}{104–112}.
\newblock
\showISBNx{9781450359016}
\urldef\tempurl%
\url{https://doi.org/10.1145/3240323.3240360}
\showDOI{\tempurl}


\bibitem[Chen et~al\mbox{.}(2021)]%
        {10.1145/3404835.3462919}
\bibfield{author}{\bibinfo{person}{Jiawei Chen}, \bibinfo{person}{Hande Dong},
  \bibinfo{person}{Yang Qiu}, \bibinfo{person}{Xiangnan He},
  \bibinfo{person}{Xin Xin}, \bibinfo{person}{Liang Chen},
  \bibinfo{person}{Guli Lin}, {and} \bibinfo{person}{Keping Yang}.}
  \bibinfo{year}{2021}\natexlab{}.
\newblock \bibinfo{booktitle}{\emph{AutoDebias: Learning to Debias for
  Recommendation}}.
\newblock \bibinfo{publisher}{Association for Computing Machinery},
  \bibinfo{address}{New York, NY, USA}, \bibinfo{pages}{21–30}.
\newblock
\showISBNx{9781450380379}
\urldef\tempurl%
\url{https://doi.org/10.1145/3404835.3462919}
\showURL{%
\tempurl}


\bibitem[Chen et~al\mbox{.}(2020)]%
        {DBLP:journals/corr/abs-2010-03240}
\bibfield{author}{\bibinfo{person}{Jiawei Chen}, \bibinfo{person}{Hande Dong},
  \bibinfo{person}{Xiang Wang}, \bibinfo{person}{Fuli Feng},
  \bibinfo{person}{Meng Wang}, {and} \bibinfo{person}{Xiangnan He}.}
  \bibinfo{year}{2020}\natexlab{}.
\newblock \showarticletitle{Bias and Debias in Recommender System: {A} Survey
  and Future Directions}.
\newblock \bibinfo{journal}{\emph{CoRR}}  \bibinfo{volume}{abs/2010.03240}
  (\bibinfo{year}{2020}).
\newblock
\showeprint[arXiv]{2010.03240}
\urldef\tempurl%
\url{https://arxiv.org/abs/2010.03240}
\showURL{%
\tempurl}


\bibitem[Cheng et~al\mbox{.}(2017)]%
        {DBLP:journals/corr/abs-1710-09282}
\bibfield{author}{\bibinfo{person}{Yu Cheng}, \bibinfo{person}{Duo Wang},
  \bibinfo{person}{Pan Zhou}, {and} \bibinfo{person}{Tao Zhang}.}
  \bibinfo{year}{2017}\natexlab{}.
\newblock \showarticletitle{A Survey of Model Compression and Acceleration for
  Deep Neural Networks}.
\newblock \bibinfo{journal}{\emph{CoRR}}  \bibinfo{volume}{abs/1710.09282}
  (\bibinfo{year}{2017}).
\newblock
\showeprint[arXiv]{1710.09282}
\urldef\tempurl%
\url{http://arxiv.org/abs/1710.09282}
\showURL{%
\tempurl}


\bibitem[Deng and Zhang(2021)]%
        {NEURIPS2021_b9f35816}
\bibfield{author}{\bibinfo{person}{Xiang Deng} {and} \bibinfo{person}{Zhongfei
  Zhang}.} \bibinfo{year}{2021}\natexlab{}.
\newblock \showarticletitle{Comprehensive Knowledge Distillation with Causal
  Intervention}. In \bibinfo{booktitle}{\emph{Advances in Neural Information
  Processing Systems}}, \bibfield{editor}{\bibinfo{person}{M.~Ranzato},
  \bibinfo{person}{A.~Beygelzimer}, \bibinfo{person}{Y.~Dauphin},
  \bibinfo{person}{P.S. Liang}, {and} \bibinfo{person}{J.~Wortman Vaughan}}
  (Eds.), Vol.~\bibinfo{volume}{34}. \bibinfo{publisher}{Curran Associates,
  Inc.}, \bibinfo{pages}{22158--22170}.
\newblock


\bibitem[Ding et~al\mbox{.}(2022)]%
        {ding2022interpolative}
\bibfield{author}{\bibinfo{person}{Sihao Ding}, \bibinfo{person}{Fuli Feng},
  \bibinfo{person}{Xiangnan He}, \bibinfo{person}{Jinqiu Jin},
  \bibinfo{person}{Wenjie Wang}, \bibinfo{person}{Yong Liao}, {and}
  \bibinfo{person}{Yongdong Zhang}.} \bibinfo{year}{2022}\natexlab{}.
\newblock \showarticletitle{Interpolative Distillation for Unifying Biased and
  Debiased Recommendation}. In \bibinfo{booktitle}{\emph{Proceedings of the
  45th International ACM SIGIR Conference on Research and Development in
  Information Retrieval}}. \bibinfo{pages}{40--49}.
\newblock


\bibitem[Feng et~al\mbox{.}(2021)]%
        {10.1145/3404835.3462971}
\bibfield{author}{\bibinfo{person}{Fuli Feng}, \bibinfo{person}{Weiran Huang},
  \bibinfo{person}{Xiangnan He}, \bibinfo{person}{Xin Xin},
  \bibinfo{person}{Qifan Wang}, {and} \bibinfo{person}{Tat-Seng Chua}.}
  \bibinfo{year}{2021}\natexlab{}.
\newblock \bibinfo{booktitle}{\emph{Should Graph Convolution Trust Neighbors? A
  Simple Causal Inference Method}}.
\newblock \bibinfo{publisher}{Association for Computing Machinery},
  \bibinfo{address}{New York, NY, USA}, \bibinfo{pages}{1208–1218}.
\newblock
\showISBNx{9781450380379}
\urldef\tempurl%
\url{https://doi.org/10.1145/3404835.3462971}
\showURL{%
\tempurl}


\bibitem[He et~al\mbox{.}(2020)]%
        {10.1145/3397271.3401063}
\bibfield{author}{\bibinfo{person}{Xiangnan He}, \bibinfo{person}{Kuan Deng},
  \bibinfo{person}{Xiang Wang}, \bibinfo{person}{Yan Li},
  \bibinfo{person}{YongDong Zhang}, {and} \bibinfo{person}{Meng Wang}.}
  \bibinfo{year}{2020}\natexlab{}.
\newblock \bibinfo{booktitle}{\emph{LightGCN: Simplifying and Powering Graph
  Convolution Network for Recommendation}}.
\newblock \bibinfo{publisher}{Association for Computing Machinery},
  \bibinfo{address}{New York, NY, USA}, \bibinfo{pages}{639–648}.
\newblock
\showISBNx{9781450380164}
\urldef\tempurl%
\url{https://doi.org/10.1145/3397271.3401063}
\showURL{%
\tempurl}


\bibitem[Hu et~al\mbox{.}(2021)]%
        {Hu_2021_CVPR}
\bibfield{author}{\bibinfo{person}{Xinting Hu}, \bibinfo{person}{Kaihua Tang},
  \bibinfo{person}{Chunyan Miao}, \bibinfo{person}{Xian-Sheng Hua}, {and}
  \bibinfo{person}{Hanwang Zhang}.} \bibinfo{year}{2021}\natexlab{}.
\newblock \showarticletitle{Distilling Causal Effect of Data in
  Class-Incremental Learning}.
\newblock \bibinfo{journal}{\emph{CVPR}}, \bibinfo{pages}{3957--3966}.
\newblock


\bibitem[Hu et~al\mbox{.}(2008)]%
        {4781121}
\bibfield{author}{\bibinfo{person}{Yifan Hu}, \bibinfo{person}{Yehuda Koren},
  {and} \bibinfo{person}{Chris Volinsky}.} \bibinfo{year}{2008}\natexlab{}.
\newblock \showarticletitle{Collaborative Filtering for Implicit Feedback
  Datasets}. In \bibinfo{booktitle}{\emph{2008 Eighth IEEE International
  Conference on Data Mining}}. \bibinfo{pages}{263--272}.
\newblock
\urldef\tempurl%
\url{https://doi.org/10.1109/ICDM.2008.22}
\showDOI{\tempurl}


\bibitem[Kang et~al\mbox{.}(2020)]%
        {10.1145/3340531.3412005}
\bibfield{author}{\bibinfo{person}{SeongKu Kang}, \bibinfo{person}{Junyoung
  Hwang}, \bibinfo{person}{Wonbin Kweon}, {and} \bibinfo{person}{Hwanjo Yu}.}
  \bibinfo{year}{2020}\natexlab{}.
\newblock \showarticletitle{DE-RRD: A Knowledge Distillation Framework for
  Recommender System} \emph{(\bibinfo{series}{CIKM '20})}.
  \bibinfo{publisher}{Association for Computing Machinery},
  \bibinfo{address}{New York, NY, USA}, \bibinfo{pages}{605–614}.
\newblock
\showISBNx{9781450368599}
\urldef\tempurl%
\url{https://doi.org/10.1145/3340531.3412005}
\showDOI{\tempurl}


\bibitem[Kang et~al\mbox{.}(2021)]%
        {10.1145/3447548.3467319}
\bibfield{author}{\bibinfo{person}{SeongKu Kang}, \bibinfo{person}{Junyoung
  Hwang}, \bibinfo{person}{Wonbin Kweon}, {and} \bibinfo{person}{Hwanjo Yu}.}
  \bibinfo{year}{2021}\natexlab{}.
\newblock \showarticletitle{Topology Distillation for Recommender System}. In
  \bibinfo{booktitle}{\emph{Proceedings of the 27th ACM SIGKDD Conference on
  Knowledge Discovery and Data Mining}} (Virtual Event, Singapore)
  \emph{(\bibinfo{series}{KDD '21})}. \bibinfo{publisher}{Association for
  Computing Machinery}, \bibinfo{address}{New York, NY, USA},
  \bibinfo{pages}{829–839}.
\newblock
\showISBNx{9781450383325}
\urldef\tempurl%
\url{https://doi.org/10.1145/3447548.3467319}
\showDOI{\tempurl}


\bibitem[Krishnan et~al\mbox{.}(2018)]%
        {10.1145/3269206.3269264}
\bibfield{author}{\bibinfo{person}{Adit Krishnan}, \bibinfo{person}{Ashish
  Sharma}, \bibinfo{person}{Aravind Sankar}, {and} \bibinfo{person}{Hari
  Sundaram}.} \bibinfo{year}{2018}\natexlab{}.
\newblock \showarticletitle{An Adversarial Approach to Improve Long-Tail
  Performance in Neural Collaborative Filtering}. In
  \bibinfo{booktitle}{\emph{Proceedings of the 27th ACM International
  Conference on Information and Knowledge Management}} (Torino, Italy)
  \emph{(\bibinfo{series}{CIKM '18})}. \bibinfo{publisher}{Association for
  Computing Machinery}, \bibinfo{address}{New York, NY, USA},
  \bibinfo{pages}{1491–1494}.
\newblock
\showISBNx{9781450360142}
\urldef\tempurl%
\url{https://doi.org/10.1145/3269206.3269264}
\showDOI{\tempurl}


\bibitem[Kweon et~al\mbox{.}(2021)]%
        {10.1145/3442381.3449878}
\bibfield{author}{\bibinfo{person}{Wonbin Kweon}, \bibinfo{person}{Seongku
  Kang}, {and} \bibinfo{person}{Hwanjo Yu}.} \bibinfo{year}{2021}\natexlab{}.
\newblock \showarticletitle{Bidirectional Distillation for Top-K Recommender
  System}. In \bibinfo{booktitle}{\emph{Proceedings of the Web Conference
  2021}} (Ljubljana, Slovenia) \emph{(\bibinfo{series}{WWW '21})}.
  \bibinfo{publisher}{Association for Computing Machinery},
  \bibinfo{address}{New York, NY, USA}, \bibinfo{pages}{3861–3871}.
\newblock
\showISBNx{9781450383127}
\urldef\tempurl%
\url{https://doi.org/10.1145/3442381.3449878}
\showDOI{\tempurl}


\bibitem[Lee et~al\mbox{.}(2019)]%
        {8970837}
\bibfield{author}{\bibinfo{person}{Jae-woong Lee}, \bibinfo{person}{Minjin
  Choi}, \bibinfo{person}{Jongwuk Lee}, {and} \bibinfo{person}{Hyunjung Shim}.}
  \bibinfo{year}{2019}\natexlab{}.
\newblock \showarticletitle{Collaborative Distillation for Top-N
  Recommendation}. In \bibinfo{booktitle}{\emph{2019 IEEE International
  Conference on Data Mining (ICDM)}}. \bibinfo{pages}{369--378}.
\newblock
\urldef\tempurl%
\url{https://doi.org/10.1109/ICDM.2019.00047}
\showDOI{\tempurl}


\bibitem[Lin et~al\mbox{.}(2019)]%
        {DBLP:journals/corr/abs-1909-06362}
\bibfield{author}{\bibinfo{person}{Kun Lin}, \bibinfo{person}{Nasim Sonboli},
  \bibinfo{person}{Bamshad Mobasher}, {and} \bibinfo{person}{Robin Burke}.}
  \bibinfo{year}{2019}\natexlab{}.
\newblock \showarticletitle{Crank up the volume: preference bias amplification
  in collaborative recommendation}.
\newblock \bibinfo{journal}{\emph{CoRR}}  \bibinfo{volume}{abs/1909.06362}
  (\bibinfo{year}{2019}).
\newblock
\showeprint[arXiv]{1909.06362}
\urldef\tempurl%
\url{http://arxiv.org/abs/1909.06362}
\showURL{%
\tempurl}


\bibitem[Liu et~al\mbox{.}(2020)]%
        {10.1145/3397271.3401083}
\bibfield{author}{\bibinfo{person}{Dugang Liu}, \bibinfo{person}{Pengxiang
  Cheng}, \bibinfo{person}{Zhenhua Dong}, \bibinfo{person}{Xiuqiang He},
  \bibinfo{person}{Weike Pan}, {and} \bibinfo{person}{Zhong Ming}.}
  \bibinfo{year}{2020}\natexlab{}.
\newblock \showarticletitle{A General Knowledge Distillation Framework for
  Counterfactual Recommendation via Uniform Data}. In
  \bibinfo{booktitle}{\emph{Proceedings of the 43rd International ACM SIGIR
  Conference on Research and Development in Information Retrieval}} (Virtual
  Event, China) \emph{(\bibinfo{series}{SIGIR '20})}.
  \bibinfo{publisher}{Association for Computing Machinery},
  \bibinfo{address}{New York, NY, USA}, \bibinfo{pages}{831–840}.
\newblock
\showISBNx{9781450380164}
\urldef\tempurl%
\url{https://doi.org/10.1145/3397271.3401083}
\showDOI{\tempurl}


\bibitem[Liu et~al\mbox{.}(2019)]%
        {Liu_2019_CVPR}
\bibfield{author}{\bibinfo{person}{Yifan Liu}, \bibinfo{person}{Ke Chen},
  \bibinfo{person}{Chris Liu}, \bibinfo{person}{Zengchang Qin},
  \bibinfo{person}{Zhenbo Luo}, {and} \bibinfo{person}{Jingdong Wang}.}
  \bibinfo{year}{2019}\natexlab{}.
\newblock \showarticletitle{Structured Knowledge Distillation for Semantic
  Segmentation}. In \bibinfo{booktitle}{\emph{Proceedings of the IEEE/CVF
  Conference on Computer Vision and Pattern Recognition (CVPR)}}.
\newblock


\bibitem[Marlin et~al\mbox{.}(2007)]%
        {10.5555/3020488.3020521}
\bibfield{author}{\bibinfo{person}{Benjamin~M. Marlin},
  \bibinfo{person}{Richard~S. Zemel}, \bibinfo{person}{Sam Roweis}, {and}
  \bibinfo{person}{Malcolm Slaney}.} \bibinfo{year}{2007}\natexlab{}.
\newblock \showarticletitle{Collaborative Filtering and the Missing at Random
  Assumption}. In \bibinfo{booktitle}{\emph{Proceedings of the Twenty-Third
  Conference on Uncertainty in Artificial Intelligence}} (Vancouver, BC,
  Canada) \emph{(\bibinfo{series}{UAI'07})}. \bibinfo{publisher}{AUAI Press},
  \bibinfo{address}{Arlington, Virginia, USA}, \bibinfo{pages}{267–275}.
\newblock
\showISBNx{0974903930}


\bibitem[Niu et~al\mbox{.}(2021)]%
        {Niu_2021_CVPR}
\bibfield{author}{\bibinfo{person}{Yulei Niu}, \bibinfo{person}{Kaihua Tang},
  \bibinfo{person}{Hanwang Zhang}, \bibinfo{person}{Zhiwu Lu},
  \bibinfo{person}{Xian-Sheng Hua}, {and} \bibinfo{person}{Ji-Rong Wen}.}
  \bibinfo{year}{2021}\natexlab{}.
\newblock \showarticletitle{Counterfactual VQA: A Cause-Effect Look at Language
  Bias}. In \bibinfo{booktitle}{\emph{Proceedings of the IEEE/CVF Conference on
  Computer Vision and Pattern Recognition (CVPR)}}.
  \bibinfo{pages}{12700--12710}.
\newblock


\bibitem[Pearl(2009)]%
        {causal:Causality}
\bibfield{author}{\bibinfo{person}{Judea Pearl}.}
  \bibinfo{year}{2009}\natexlab{}.
\newblock \bibinfo{booktitle}{\emph{Causality}}.
\newblock \bibinfo{publisher}{Cambridge University Press}.
\newblock


\bibitem[Rendle et~al\mbox{.}(2012)]%
        {DBLP:journals/corr/abs-1205-2618}
\bibfield{author}{\bibinfo{person}{Steffen Rendle}, \bibinfo{person}{Christoph
  Freudenthaler}, \bibinfo{person}{Zeno Gantner}, {and} \bibinfo{person}{Lars
  Schmidt{-}Thieme}.} \bibinfo{year}{2012}\natexlab{}.
\newblock \showarticletitle{{BPR:} Bayesian Personalized Ranking from Implicit
  Feedback}.
\newblock \bibinfo{journal}{\emph{CoRR}}  \bibinfo{volume}{abs/1205.2618}
  (\bibinfo{year}{2012}).
\newblock
\showeprint[arXiv]{1205.2618}
\urldef\tempurl%
\url{http://arxiv.org/abs/1205.2618}
\showURL{%
\tempurl}


\bibitem[Saito(2019)]%
        {DBLP:journals/corr/abs-1910-01444}
\bibfield{author}{\bibinfo{person}{Yuta Saito}.}
  \bibinfo{year}{2019}\natexlab{}.
\newblock \showarticletitle{Eliminating Bias in Recommender Systems via
  Pseudo-Labeling}.
\newblock \bibinfo{journal}{\emph{CoRR}}  \bibinfo{volume}{abs/1910.01444}
  (\bibinfo{year}{2019}).
\newblock
\showeprint[arXiv]{1910.01444}
\urldef\tempurl%
\url{http://arxiv.org/abs/1910.01444}
\showURL{%
\tempurl}


\bibitem[Schnabel et~al\mbox{.}(2016)]%
        {pmlr-v48-schnabel16}
\bibfield{author}{\bibinfo{person}{Tobias Schnabel}, \bibinfo{person}{Adith
  Swaminathan}, \bibinfo{person}{Ashudeep Singh}, \bibinfo{person}{Navin
  Chandak}, {and} \bibinfo{person}{Thorsten Joachims}.}
  \bibinfo{year}{2016}\natexlab{}.
\newblock \showarticletitle{Recommendations as Treatments: Debiasing Learning
  and Evaluation}. In \bibinfo{booktitle}{\emph{Proceedings of The 33rd
  International Conference on Machine Learning}}
  \emph{(\bibinfo{series}{Proceedings of Machine Learning Research},
  Vol.~\bibinfo{volume}{48})}, \bibfield{editor}{\bibinfo{person}{Maria~Florina
  Balcan} {and} \bibinfo{person}{Kilian~Q. Weinberger}} (Eds.).
  \bibinfo{publisher}{PMLR}, \bibinfo{address}{New York, New York, USA},
  \bibinfo{pages}{1670--1679}.
\newblock
\urldef\tempurl%
\url{https://proceedings.mlr.press/v48/schnabel16.html}
\showURL{%
\tempurl}


\bibitem[Tang and Wang(2018)]%
        {10.1145/3219819.3220021}
\bibfield{author}{\bibinfo{person}{Jiaxi Tang} {and} \bibinfo{person}{Ke
  Wang}.} \bibinfo{year}{2018}\natexlab{}.
\newblock \showarticletitle{Ranking Distillation: Learning Compact Ranking
  Models With High Performance for Recommender System}. In
  \bibinfo{booktitle}{\emph{Proceedings of the 24th ACM SIGKDD International
  Conference on Knowledge Discovery and Data Mining}} (London, United Kingdom)
  \emph{(\bibinfo{series}{KDD'18})}. \bibinfo{publisher}{Association for
  Computing Machinery}, \bibinfo{address}{New York, NY, USA},
  \bibinfo{pages}{2289–2298}.
\newblock
\showISBNx{9781450355520}
\urldef\tempurl%
\url{https://doi.org/10.1145/3219819.3220021}
\showDOI{\tempurl}


\bibitem[Tang et~al\mbox{.}(2020)]%
        {tang2020longtailed}
\bibfield{author}{\bibinfo{person}{Kaihua Tang}, \bibinfo{person}{Jianqiang
  Huang}, {and} \bibinfo{person}{Hanwang Zhang}.}
  \bibinfo{year}{2020}\natexlab{}.
\newblock \showarticletitle{Long-Tailed Classification by Keeping the Good and
  Removing the Bad Momentum Causal Effect}. In
  \bibinfo{booktitle}{\emph{NeurIPS}}.
\newblock


\bibitem[Tao et~al\mbox{.}(2022)]%
        {10.1145/3485447.3512003}
\bibfield{author}{\bibinfo{person}{Ye Tao}, \bibinfo{person}{Ying Li},
  \bibinfo{person}{Su Zhang}, \bibinfo{person}{Zhirong Hou}, {and}
  \bibinfo{person}{Zhonghai Wu}.} \bibinfo{year}{2022}\natexlab{}.
\newblock \showarticletitle{Revisiting Graph Based Social Recommendation: A
  Distillation Enhanced Social Graph Network}. In
  \bibinfo{booktitle}{\emph{Proceedings of the ACM Web Conference 2022}}
  (Virtual Event, Lyon, France) \emph{(\bibinfo{series}{WWW '22})}.
  \bibinfo{publisher}{Association for Computing Machinery},
  \bibinfo{address}{New York, NY, USA}, \bibinfo{pages}{2830–2838}.
\newblock
\showISBNx{9781450390965}
\urldef\tempurl%
\url{https://doi.org/10.1145/3485447.3512003}
\showDOI{\tempurl}


\bibitem[Wang et~al\mbox{.}(2021)]%
        {10.1145/3459637.3482117}
\bibfield{author}{\bibinfo{person}{Yuening Wang}, \bibinfo{person}{Yingxue
  Zhang}, {and} \bibinfo{person}{Mark Coates}.}
  \bibinfo{year}{2021}\natexlab{}.
\newblock \showarticletitle{Graph Structure Aware Contrastive Knowledge
  Distillation for Incremental Learning in Recommender Systems}. In
  \bibinfo{booktitle}{\emph{Proceedings of the 30th ACM International
  Conference on Information Knowledge Management}} (Virtual Event, Queensland,
  Australia) \emph{(\bibinfo{series}{CIKM '21})}.
  \bibinfo{publisher}{Association for Computing Machinery},
  \bibinfo{address}{New York, NY, USA}, \bibinfo{pages}{3518–3522}.
\newblock
\showISBNx{9781450384469}
\urldef\tempurl%
\url{https://doi.org/10.1145/3459637.3482117}
\showDOI{\tempurl}


\bibitem[Wei et~al\mbox{.}(2021)]%
        {10.1145/3447548.3467289}
\bibfield{author}{\bibinfo{person}{Tianxin Wei}, \bibinfo{person}{Fuli Feng},
  \bibinfo{person}{Jiawei Chen}, \bibinfo{person}{Ziwei Wu},
  \bibinfo{person}{Jinfeng Yi}, {and} \bibinfo{person}{Xiangnan He}.}
  \bibinfo{year}{2021}\natexlab{}.
\newblock \showarticletitle{Model-Agnostic Counterfactual Reasoning for
  Eliminating Popularity Bias in Recommender System}. In
  \bibinfo{booktitle}{\emph{Proceedings of the 27th ACM SIGKDD Conference on
  Knowledge Discovery and Data Mining}} (Virtual Event, Singapore)
  \emph{(\bibinfo{series}{KDD '21})}. \bibinfo{publisher}{Association for
  Computing Machinery}, \bibinfo{address}{New York, NY, USA},
  \bibinfo{pages}{1791–1800}.
\newblock
\showISBNx{9781450383325}
\urldef\tempurl%
\url{https://doi.org/10.1145/3447548.3467289}
\showDOI{\tempurl}


\bibitem[Xia et~al\mbox{.}(2022)]%
        {xia2022device}
\bibfield{author}{\bibinfo{person}{Xin Xia}, \bibinfo{person}{Hongzhi Yin},
  \bibinfo{person}{Junliang Yu}, \bibinfo{person}{Qinyong Wang},
  \bibinfo{person}{Guandong Xu}, {and} \bibinfo{person}{Quoc Viet~Hung
  Nguyen}.} \bibinfo{year}{2022}\natexlab{}.
\newblock \showarticletitle{On-Device Next-Item Recommendation with
  Self-Supervised Knowledge Distillation}. In
  \bibinfo{booktitle}{\emph{Proceedings of the 45th International ACM SIGIR
  Conference on Research and Development in Information Retrieval}}.
  \bibinfo{pages}{546--555}.
\newblock


\bibitem[Xu et~al\mbox{.}(2021)]%
        {xu2021learning}
\bibfield{author}{\bibinfo{person}{Shuyuan Xu}, \bibinfo{person}{Yunqi Li},
  \bibinfo{person}{Shuchang Liu}, \bibinfo{person}{Zuohui Fu},
  \bibinfo{person}{Yingqiang Ge}, \bibinfo{person}{Xu Chen}, {and}
  \bibinfo{person}{Yongfeng Zhang}.} \bibinfo{year}{2021}\natexlab{}.
\newblock \showarticletitle{Learning causal explanations for recommendation}.
  In \bibinfo{booktitle}{\emph{The 1st International Workshop on Causality in
  Search and Recommendation}}.
\newblock


\bibitem[Xu et~al\mbox{.}(2022)]%
        {DBLP:journals/corr/abs-2201-08024}
\bibfield{author}{\bibinfo{person}{Zixuan Xu}, \bibinfo{person}{Penghui Wei},
  \bibinfo{person}{Weimin Zhang}, \bibinfo{person}{Shaoguo Liu},
  \bibinfo{person}{Liang Wang}, {and} \bibinfo{person}{Bo Zheng}.}
  \bibinfo{year}{2022}\natexlab{}.
\newblock \showarticletitle{{UKD:} Debiasing Conversion Rate Estimation via
  Uncertainty-regularized Knowledge Distillation}.
\newblock \bibinfo{journal}{\emph{CoRR}}  \bibinfo{volume}{abs/2201.08024}
  (\bibinfo{year}{2022}).
\newblock
\showeprint[arXiv]{2201.08024}
\urldef\tempurl%
\url{https://arxiv.org/abs/2201.08024}
\showURL{%
\tempurl}


\bibitem[Yang et~al\mbox{.}(2022)]%
        {2022Cross}
\bibfield{author}{\bibinfo{person}{C. Yang}, \bibinfo{person}{J. Pan},
  \bibinfo{person}{X. Gao}, \bibinfo{person}{T. Jiang}, \bibinfo{person}{D.
  Liu}, {and} \bibinfo{person}{G. Chen}.} \bibinfo{year}{2022}\natexlab{}.
\newblock \showarticletitle{Cross-Task Knowledge Distillation in Multi-Task
  Recommendation}.
\newblock  (\bibinfo{year}{2022}).
\newblock


\bibitem[Yang et~al\mbox{.}(2018)]%
        {10.1145/3240323.3240355}
\bibfield{author}{\bibinfo{person}{Longqi Yang}, \bibinfo{person}{Yin Cui},
  \bibinfo{person}{Yuan Xuan}, \bibinfo{person}{Chenyang Wang},
  \bibinfo{person}{Serge Belongie}, {and} \bibinfo{person}{Deborah Estrin}.}
  \bibinfo{year}{2018}\natexlab{}.
\newblock \showarticletitle{Unbiased Offline Recommender Evaluation for
  Missing-Not-at-Random Implicit Feedback}. In
  \bibinfo{booktitle}{\emph{Proceedings of the 12th ACM Conference on
  Recommender Systems}} (Vancouver, British Columbia, Canada)
  \emph{(\bibinfo{series}{RecSys '18})}. \bibinfo{publisher}{Association for
  Computing Machinery}, \bibinfo{address}{New York, NY, USA},
  \bibinfo{pages}{279–287}.
\newblock
\showISBNx{9781450359016}
\urldef\tempurl%
\url{https://doi.org/10.1145/3240323.3240355}
\showDOI{\tempurl}


\bibitem[Yang et~al\mbox{.}(2021)]%
        {10.1145/3459637.3482305}
\bibfield{author}{\bibinfo{person}{Mengyue Yang}, \bibinfo{person}{Quanyu Dai},
  \bibinfo{person}{Zhenhua Dong}, \bibinfo{person}{Xu Chen},
  \bibinfo{person}{Xiuqiang He}, {and} \bibinfo{person}{Jun Wang}.}
  \bibinfo{year}{2021}\natexlab{}.
\newblock \showarticletitle{Top-N Recommendation with Counterfactual User
  Preference Simulation}. In \bibinfo{booktitle}{\emph{Proceedings of the 30th
  ACM International Conference on Information Knowledge Management}} (Virtual
  Event, Queensland, Australia) \emph{(\bibinfo{series}{CIKM '21})}.
  \bibinfo{publisher}{Association for Computing Machinery},
  \bibinfo{address}{New York, NY, USA}, \bibinfo{pages}{2342–2351}.
\newblock
\showISBNx{9781450384469}
\urldef\tempurl%
\url{https://doi.org/10.1145/3459637.3482305}
\showDOI{\tempurl}


\bibitem[Yim et~al\mbox{.}(2017)]%
        {Yim_2017_CVPR}
\bibfield{author}{\bibinfo{person}{Junho Yim}, \bibinfo{person}{Donggyu Joo},
  \bibinfo{person}{Jihoon Bae}, {and} \bibinfo{person}{Junmo Kim}.}
  \bibinfo{year}{2017}\natexlab{}.
\newblock \showarticletitle{A Gift From Knowledge Distillation: Fast
  Optimization, Network Minimization and Transfer Learning}. In
  \bibinfo{booktitle}{\emph{Proceedings of the IEEE Conference on Computer
  Vision and Pattern Recognition (CVPR)}}.
\newblock


\bibitem[Zhang et~al\mbox{.}(2016)]%
        {10.1145/2911451.2911502}
\bibfield{author}{\bibinfo{person}{Hanwang Zhang}, \bibinfo{person}{Fumin
  Shen}, \bibinfo{person}{Wei Liu}, \bibinfo{person}{Xiangnan He},
  \bibinfo{person}{Huanbo Luan}, {and} \bibinfo{person}{Tat-Seng Chua}.}
  \bibinfo{year}{2016}\natexlab{}.
\newblock \showarticletitle{Discrete Collaborative Filtering}. In
  \bibinfo{booktitle}{\emph{Proceedings of the 39th International ACM SIGIR
  Conference on Research and Development in Information Retrieval}} (Pisa,
  Italy) \emph{(\bibinfo{series}{SIGIR '16})}. \bibinfo{publisher}{Association
  for Computing Machinery}, \bibinfo{pages}{325–334}.
\newblock
\showISBNx{9781450340694}


\bibitem[Zhang et~al\mbox{.}(2021)]%
        {10.1145/3404835.3462875}
\bibfield{author}{\bibinfo{person}{Yang Zhang}, \bibinfo{person}{Fuli Feng},
  \bibinfo{person}{Xiangnan He}, \bibinfo{person}{Tianxin Wei},
  \bibinfo{person}{Chonggang Song}, \bibinfo{person}{Guohui Ling}, {and}
  \bibinfo{person}{Yongdong Zhang}.} \bibinfo{year}{2021}\natexlab{}.
\newblock \bibinfo{booktitle}{\emph{Causal Intervention for Leveraging
  Popularity Bias in Recommendation}}.
\newblock \bibinfo{publisher}{Association for Computing Machinery},
  \bibinfo{address}{New York, NY, USA}, \bibinfo{pages}{11–20}.
\newblock
\showISBNx{9781450380379}
\urldef\tempurl%
\url{https://doi.org/10.1145/3404835.3462875}
\showURL{%
\tempurl}


\bibitem[Zhang et~al\mbox{.}(2020)]%
        {zhang2020distilling}
\bibfield{author}{\bibinfo{person}{Yuan Zhang}, \bibinfo{person}{Xiaoran Xu},
  \bibinfo{person}{Hanning Zhou}, {and} \bibinfo{person}{Yan Zhang}.}
  \bibinfo{year}{2020}\natexlab{}.
\newblock \showarticletitle{Distilling structured knowledge into embeddings for
  explainable and accurate recommendation}. In
  \bibinfo{booktitle}{\emph{Proceedings of the 13th International Conference on
  Web Search and Data Mining}}. \bibinfo{pages}{735--743}.
\newblock


\bibitem[Zhu et~al\mbox{.}(2020)]%
        {zhu2020ensembled}
\bibfield{author}{\bibinfo{person}{Jieming Zhu}, \bibinfo{person}{Jinyang Liu},
  \bibinfo{person}{Weiqi Li}, \bibinfo{person}{Jincai Lai},
  \bibinfo{person}{Xiuqiang He}, \bibinfo{person}{Liang Chen}, {and}
  \bibinfo{person}{Zibin Zheng}.} \bibinfo{year}{2020}\natexlab{}.
\newblock \showarticletitle{Ensembled CTR prediction via knowledge
  distillation}. In \bibinfo{booktitle}{\emph{Proceedings of the 29th ACM
  International Conference on Information Knowledge Management}}.
  \bibinfo{pages}{2941--2958}.
\newblock


\bibitem[Zhu et~al\mbox{.}(2021)]%
        {10.1145/3437963.3441820}
\bibfield{author}{\bibinfo{person}{Ziwei Zhu}, \bibinfo{person}{Yun He},
  \bibinfo{person}{Xing Zhao}, \bibinfo{person}{Yin Zhang},
  \bibinfo{person}{Jianling Wang}, {and} \bibinfo{person}{James Caverlee}.}
  \bibinfo{year}{2021}\natexlab{}.
\newblock \showarticletitle{Popularity-Opportunity Bias in Collaborative
  Filtering}. In \bibinfo{booktitle}{\emph{Proceedings of the 14th ACM
  International Conference on Web Search and Data Mining}} (Virtual Event,
  Israel) \emph{(\bibinfo{series}{WSDM '21})}. \bibinfo{publisher}{Association
  for Computing Machinery}, \bibinfo{address}{New York, NY, USA},
  \bibinfo{pages}{85–93}.
\newblock
\showISBNx{9781450382977}
\urldef\tempurl%
\url{https://doi.org/10.1145/3437963.3441820}
\showDOI{\tempurl}


\end{thebibliography}


\end{document}